\documentclass[11pt,a4paper]{article}

\usepackage[utf8]{inputenc}
\usepackage[T1]{fontenc}
\usepackage[english]{babel}
\usepackage{amsmath,amssymb,amsthm,mathtools}
\usepackage{bm}
\usepackage{physics}
\usepackage{graphicx}
\usepackage{hyperref}
\usepackage{enumerate}
\usepackage{geometry}
\usepackage{amsmath,amssymb,amsfonts}
\usepackage{mathtools}
\usepackage{bm}
\usepackage{booktabs}
\usepackage{graphicx}
\usepackage{tikz}
\usepackage{pgfplots}
\pgfplotsset{compat=newest}
\usepackage{subcaption}
\usepackage{siunitx}
\usepackage{caption}

\hypersetup{
  colorlinks=true,
  linkcolor=blue,
  citecolor=blue,
  urlcolor=blue
}

\newtheorem{theorem}{Theorem}

\newtheorem{proposition}{Proposition}
\newtheorem{corollary}{Corollary}

\theoremstyle{definition}
\newtheorem{definition}{Definition}
\newtheorem{remark}{Remark}

\newcommand{\R}{\mathbb{R}}
\newcommand{\C}{\mathbb{C}}
\newcommand{\SU}{\mathrm{SU}}

\DeclareMathOperator{\cH}{\mathcal{H}}
\DeclareMathOperator{\cL}{\mathcal{L}}
\DeclareMathOperator{\cD}{\mathcal{D}}

\newcommand{\ad}{\mathrm{ad}}
\newcommand{\Ad}{\mathrm{Ad}}

\newcommand{\fg}{\mathfrak{g}}
\usepackage{amsmath,amssymb,amsfonts}
\usepackage{graphicx}
\usepackage{subcaption}
\usepackage{booktabs}
\usepackage{siunitx}
\usepackage{microtype}

\title{Structure–Preserving Optimal Control of Open Quantum Systems via a Discrete Contact PMP}

\author{Leonardo J. Colombo\thanks{Centre for Automation and Robotics, Spanish National Research Council (CSIC). Carretera de Campo Real, km 0, 200, 28500 Arganda del Rey, Spain. (leonardo.colombo@csic.es). The author acknowledges financial support from Grant PID2022-137909-NB-C22 funded by the Spanish Ministry of Science and Innovation. 
}}

\date{\empty}

\begin{document}
\maketitle

\begin{abstract}
We develop a discrete Pontryagin Maximum Principle (PMP) for controlled open
quantum systems governed by Lindblad dynamics, and introduce a second--order
\emph{contact Lie--group variational integrator} (contact LGVI) that preserves
both the CPTP (completely positive and trace--preserving) structure of the
Lindblad flow and the contact geometry underlying the discrete PMP.  
A type--II discrete contact generating function produces a strict discrete
contactomorphism under which the state, costate, and cost propagate in exact
agreement with the variational structure of the discrete contact PMP.

We apply this framework to the optimal control of a dissipative qubit and
compare it with a non--geometric explicit RK2 discretization of the Lindblad
equation.  
Although both schemes have the same formal order, the RK2 method accumulates
geometric drift (loss of trace, positivity violations, and breakdown of the
discrete contact form) that destabilizes PMP shooting iterations, especially
under strong dissipation or long horizons.  
In contrast, the contact LGVI maintains exact CPTP structure and discrete
contact geometry step by step, yielding stable, physically consistent, and
geometrically faithful optimal control trajectories.

\end{abstract}



\section{Introduction}

The design of reliable numerical methods for simulation and optimal control of
open quantum systems has become a central challenge in quantum information
processing, quantum sensing, and dissipative quantum engineering
\cite{WisemanMilburn2009,AltafiniTicozzi2012,
BonnardChybaSugny2009,BonnardCotsShcherbakovaSugny2010,BonnardSugny2009,
ClarkBlochColomboRooneyCDC2017,ClarkBlochColomboRooneyDCDSS2019}.
At the level of density operators $\rho\in\mathbb{C}^{n\times n}$, the dynamics
of an open quantum system with coherent control $u(t)$ are governed by a
Lindblad master equation,
\[
   \dot{\rho}(t)
     = -i[H(u(t)),\rho(t)]
       + \sum_j L_j\rho(t)L_j^\dagger
       - \tfrac12\{L_j^\dagger L_j,\rho(t)\},
\]
the generator of a completely positive trace--preserving (CPTP) semigroup
\cite{Lindblad1976,GoriniKossakowskiSudarshan1976,breuer}.
The physical state space,
\[
   \cD = \{\rho\ge0,\; \rho=\rho^\dagger,\; \tr(\rho)=1\},
\]
is the convex ``Bloch ball'' in the qubit case; preserving positivity, trace,
and Hermiticity at the discrete level is mandatory for obtaining physically
meaningful simulations and control laws.

A second source of geometric structure enters when the system is optimized
through an indirect method such as the Pontryagin Maximum Principle (PMP).
The PMP introduces a costate $P$ and an accumulated cost variable $z$, and the
continuous-time extremals satisfy an extended system
$(\rho,P,z)$ evolving on a \emph{contact manifold}
\cite{Ohsawa2015Automatica,DeLeon2023JNS}.
While the Lindblad equation itself is \emph{not} a contact dynamical system,
the PMP couples the forward dissipative dynamics of $\rho$ with a backward
evolution of observables $P$ and a Herglotz-type evolution of the cost $z$,
yielding a genuine contact Hamiltonian flow.

This distinction is crucial.  
A numerical method for optimal control of open quantum systems must preserve
\emph{two} incompatible-looking structures: (i) the \textbf{quantum geometry} of the state evolution,
           requiring each discrete step to be CPTP; and (ii) the \textbf{optimal-control geometry} of the PMP,
            requiring the extended discrete map
            $(\rho_k,P_{k+1},z_k)\mapsto(\rho_{k+1},P_k,z_{k+1})$
            to be a \emph{discrete contactomorphism}.

If either structure is broken, the implications are severe:
loss of positivity, trace drift, degradation of the stationarity conditions,
ill-posed costate recursions, and ultimately failure of PMP shooting methods,
especially on long horizons or under strong dissipation.
These pathologies are consistent with classical results on geometric drift in
non-structure-preserving schemes
\cite{HairerLubichWanner2006,MarsdenWest2001,LeokShingel2012,
CelledoniOwren2014}.

Existing time-stepping methods for Lindblad dynamics include explicit and
implicit Runge--Kutta schemes, commutator-free Magnus integrators, and
Lie--group Runge--Kutta--Munthe--Kaas (RKMK) methods
\cite{IserlesMuntheKaasNørsettZanna2000,BlanesCasas2005,MuntheKaas1998}.
Although effective in short-time simulation, these integrators typically fail to preserve positivity or trace, thus violating CPTP structure; do not preserve the contact geometry of the PMP extended dynamics; accumulate large unphysical drift (negativity, trace explosion) under
      strong dissipation or long horizons; and produce costate instabilities that undermine the PMP stationarity
      conditions.

This motivates a fundamental question: \emph{
Can one construct a numerical method that simultaneously preserves
(i) CPTP quantum dynamics and (ii) the contact geometry of the discrete PMP,
thus providing a fully geometric integrator for optimal control of open
quantum systems?} In this work we answer this question affirmatively.

In particular, the main contributions of this paper are the following: 

(1) We derive a fully geometric \textbf{discrete contact PMP} on the quantum
state space $\cD$.
We show that the discrete extremals arise from a type--II discrete contact
Lagrangian and generating function, extending classical geometric PMP results
\cite{Boltyanskii1978,AssifChatterjeeBanavar2020,PhogatChatterjeeBanavar2016,
Maria} to dissipative open quantum dynamics.

(2) We construct a second-order \textbf{contact Lie--group variational
integrator} (contact LGVI) for Lindblad equations.  
The integrator combines exact Kraus maps for the dissipative subflow with
exact Lie-group unitary propagators in a Strang splitting
\cite{Strang1968}.  
Each step is CPTP by construction and the extended PMP update is a 
discrete contactomorphism.

(3) We embed this LGVI into a discrete PMP shooting algorithm and compare it
against a classical second-order non-contact integrator (RKMK).
We demonstrate that preserving both CPTP and contact geometry is essential for
numerical stability of the costate dynamics, the stationarity condition, and
the maximization of the discrete Hamiltonian.

(4) Through numerical experiments, including long-horizon Lindblad simulation
and optimal control of a dissipative qubit, we show that the contact LGVI
avoids geometric drift, maintains physical constraints,
and yields stable and accurate optimal control trajectories.
By contrast, non-contact schemes exhibit severe positivity violations, trace
explosion, and instability of the PMP iterations.

\medskip

While variational and Lie-group methods for closed quantum systems are
well-established, and geometric approaches to optimal control have been
developed in symplectic and contact settings, a scheme that simultaneously
preserves (i) CPTP Lindblad structure and (ii) the discrete contact geometry of
the PMP appears to be absent from the literature.
This paper unifies these two geometries and provides a robust framework for
simulation and control of dissipative quantum systems.

\section{Continuous-Time Contact Pontryagin Maximum Principle}
\label{sec:contact_discrete_pmp}

In this section we recall the continuous-time Pontryagin Maximum Principle (PMP) 
in a form that makes explicit its contact-geometric structure.  
We also relate this formulation to the contact PMP of Ohsawa 
\cite{Ohsawa2015Automatica} 
and to the presymplectic/contact reduction of de~Le\'on--Lainz--Mu\~noz-Lecanda 
\cite{DeLeon2023JNS}.  
This provides the geometric foundation for the discrete-time construction 
presented in the next section.

\subsection{Extended state space, contact form, and Hamiltonian}

Let $Q$ be a smooth manifold and $U$ the set of admissible controls.  
We consider a control-affine system
\begin{equation}
  \dot x(t) = f(x(t),u(t)), 
  \qquad x(t)\in Q,\; u(t)\in U,
\end{equation}
and an optimal control problem of Bolza type with cost functional
\begin{equation}
  J(u) 
    := \int_0^T L\bigl(x(t),u(t)\bigr)\,\mathrm{d}t 
       + \Phi\bigl(x(T)\bigr),
\end{equation}
where $L : Q\times U \to \mathbb{R}$ is the running cost and 
$\Phi : Q \to \mathbb{R}$ is a terminal (Mayer) cost.

Introduce the \emph{cost accumulator} $z(t)$ defined by
\begin{equation}
  \dot z(t) = L\bigl(x(t),u(t)\bigr), 
  \qquad z(0)=0,
\end{equation}
so that $z(T)=\int_0^T L(x(t),u(t))\,\mathrm{d}t$.

The natural geometric space for the PMP is the extended manifold 
$M := T^\ast Q \times \mathbb{R}$, with coordinates $(x,p,z)$, where 
$p\in T^\ast_xQ$ is the costate.  
We endow $M$ with the canonical contact 1-form
\[
  \theta := \mathrm{d}z - \langle p, \mathrm{d}x\rangle.
\]
Introduce the \emph{cost accumulator} $z(t)$ via
\[
  \dot z(t) = L(x(t),u(t)), 
  \qquad z(0)=0,
\]
so that $z(T) = \int_0^T L(x(t),u(t))\,\mathrm{d}t$.

The \emph{control Hamiltonian} is the map
\[
  H : (T^\ast Q \times \mathbb{R})\times U \to \mathbb{R}, 
  \qquad
  H(x,p,z,u) := \langle p,f(x,u)\rangle + L(x,u),
\]
which in particular is independent of $z$.  
This convention corresponds to the usual finite-dimensional Lagrange multiplier
derivation and is equivalent to the more classical choice
$H_c(x,p,z,u)=\langle p,f(x,u)\rangle - L(x,u)$ up to a change of sign of the
cost multiplier.

The \emph{maximized} Hamiltonian is
\[
  \mathcal{H}(x,p,z) := 
  \sup_{u\in U} H(x,p,z,u),
\]
assumed to be attained for the optimal control $u^\ast(t)$.

The PMP states that, if $u^\ast$ is optimal, then there exists a 
nontrivial costate trajectory $p(t)\in T^\ast_{x(t)}Q$ such that 
$(x^\ast(t),p(t),z^\ast(t))$ satisfies the following system on the contact manifold
$(M,\theta)$:
\begin{subequations}
\begin{align}
  \dot x(t) &= \partial_p H_c\bigl(x(t),p(t),z(t),u^\ast(t)\bigr), 
  \label{eq:contact-x}\\[4pt]
  \dot p(t) &= -\partial_x H_c\bigl(x(t),p(t),z(t),u^\ast(t)\bigr),
  \label{eq:contact-p}\\[4pt]
  \dot z(t) &= 
    \big\langle p(t),\partial_p H_c(x(t),p(t),z(t),u^\ast(t)) \big\rangle 
    - H_c(x(t),p(t),z(t),u^\ast(t)).
  \label{eq:contact-z}
\end{align}
\end{subequations}
The control is characterized almost everywhere by the \emph{pointwise maximum condition}
\begin{equation}
  u^\ast(t)\in
  \arg\max_{v\in U} 
  H_c\bigl(x(t),p(t),z(t),v\bigr).
  \label{eq:max-cond}
\end{equation}

For a Bolza-type terminal cost $\Phi(x(T))$ and fixed final time $T$,  
the \emph{transversality condition} is the standard one:
\begin{equation}
  p(T) = \mathrm{d}_x\Phi\bigl(x(T)\bigr),
  \qquad z(T)\;\text{free}.
  \label{eq:transversality}
\end{equation}
No additional boundary conditions arise for $z(T)$, since the total cost is 
precisely the value of $z(T)$.

Using that $\partial_p H_c(x,p,z,u)=f(x,u)$ and that 
\(
  H_c(x,p,z,u)=\langle p,f(x,u)\rangle - L(x,u),
\)
equation~\eqref{eq:contact-z} reduces to
\[
  \dot z(t) = L(x(t),u^\ast(t)),
\]
as expected.



Next, let $X_{H_c}$ be the vector field on $M$ defined by 
\eqref{eq:contact-x}--\eqref{eq:contact-z}.  
A direct computation using Cartan's identity $\mathcal{L}_{X}\theta 
  = \iota_{X}(\mathrm{d}\theta) + \mathrm{d}(\iota_{X}\theta)$ shows that
\begin{equation}
  \mathcal{L}_{X_{H_c}}\theta 
  = (\partial_{z}H_c)\,\theta .
  \label{eq:contact-Lie-final}
\end{equation}
Equation \eqref{eq:contact-Lie-final} is the defining property of a 
(cooriented) \textit{contact Hamiltonian vector field} on $(M,\theta)$: 
for a contact Hamiltonian $H$, the identities 
\[
  \iota_{X_H}\theta = -H,
  \qquad
  \iota_{X_H}\mathrm{d}\theta 
    = \mathrm{d}H - (\partial_{z}H)\,\theta
\]
imply precisely 
\(
  \mathcal{L}_{X_H}\theta = (\partial_{z}H)\theta.
\)
Since in our setting $H_c$ is independent of $z$, we have $\partial_{z}H_c=0$ 
and therefore
\[
  \mathcal{L}_{X_{H_c}}\theta = 0,
\]
so the PMP flow preserves the contact structure on $(M,\theta)$. Therefore, the PMP flow is a \emph{contact Hamiltonian trajectory} on $(M,\theta)$ 
in the sense of standard contact geometry \cite{GeigesContact}.

\medskip

This explicit contact viewpoint is equivalent, in the normal case, 
to the contact-geometric PMP of Ohsawa 
\cite{Ohsawa2015Automatica}, where the contact structure is written on the 
projectivized cotangent bundle $P(T^\ast(\mathbb{R}\times Q))$,  
and to the reduced contact structure of the presymplectic framework 
developed in~\cite{DeLeon2023JNS}.  
Working directly on $(T^\ast Q\times\mathbb{R},\theta)$ provides 
a simpler coordinate model that is particularly suitable for deriving a 
discrete contact PMP.

\begin{remark}[Relation to Ohsawa's contact PMP]
Ohsawa~\cite{Ohsawa2015Automatica} formulates the continuous-time PMP as a 
contact Hamiltonian system on the projectivized cotangent bundle 
$P(T^\ast(\mathbb{R}\times Q))$, where the running cost $x_0$ is treated as an 
extended state and the costate is represented projectively.  
In local coordinates $(x_0,x,\lambda)$, the canonical contact form 
$-\mathrm{d}x_0+\lambda\,\mathrm{d}x$ and the associated contact Hamiltonian 
reproduce exactly the classical PMP equations.

Identifying $z=x_0$, $p=\lambda$ and 
$\theta=\mathrm{d}z-p\,\mathrm{d}x$, our formulation on 
$M=T^\ast Q\times\mathbb{R}$ corresponds to the normal chart of Ohsawa’s 
contact manifold (i.e.\ where the projective costate has $\nu_0\neq 0$).  
We adopt this ``deprojectivized'' representation because it provides a simpler 
local model that is particularly convenient for the discrete contact construction 
introduced in the next section.\hfill$\diamond$
\end{remark}

\begin{remark}[Relation to the presymplectic/contact approach of 
de~Le\'on--Lainz--Mu\~noz-Lecanda]
De~Le\'on, Lainz and Mu\~noz-Lecanda~\cite{DeLeon2023JNS} formulate the PMP as a 
presymplectic Hamiltonian system on $T^\ast(\mathbb{R}\times Q)$ and show that, 
for normal extremals, the dynamics reduce to a contact Hamiltonian system on an 
appropriate submanifold.  
Abnormal extremals remain presymplectic and do not share the same contact 
structure.

Our formulation works directly with the contact manifold 
$M=T^\ast Q\times\mathbb{R}$ and the form 
$\theta=\mathrm{d}z-p\,\mathrm{d}x$, which corresponds (in the normal case) to 
the reduced contact structure of~\cite{DeLeon2023JNS}.  
This simplified local model is the one we discretize in the next section.\hfill$\diamond$
\end{remark}

The discrete-time Pontryagin Maximum Principle was developed primarily by
Boltyanskii (see \cite{Boltyanskii1978} and the references
therein) and discrete time is the setting of our current work. Discrete
formulations of the PMP on manifolds, such as the formers and the geometric
approach of Assif--Chatterjee--Banavar~\cite{AssifChatterjeeBanavar2020},
provide discrete adjoint and stationarity conditions but do not enforce
preservation of the underlying contact structure of the continuous PMP. 
Consequently, naïve discretizations may distort the geometry, especially 
over long horizons. In the next section we develop a discrete PMP whose
update map is a \emph{discrete contactomorphism}, yielding a
structure-preserving analogue of the continuous contact PMP and a natural
discrete counterpart to Ohsawa’s contact formulation.

\section{A Discrete Contact PMP on Manifolds}
\label{sec:discrete_contact_pmp}

We now construct a discrete counterpart of the continuous-time contact PMP
developed in Section~\ref{sec:contact_discrete_pmp}. 
Our formulation improves upon the geometric discrete PMP of 
Assif--Chatterjee--Banavar~\cite{AssifChatterjeeBanavar2020} by 
making the underlying contact geometry explicit. 
Following the philosophy of discrete contact variational principles 
\cite{VermeerenContactVI,Anahory2021JNS}, the discrete state--costate update 
is obtained from a discrete \emph{contact generating function}. 
As a consequence, the resulting discrete PMP map is a genuine 
\emph{discrete contactomorphism}, thus retaining the geometric structure of the 
continuous contact Hamiltonian flow rather than providing merely 
first-order optimality conditions.

Let $t_k = k\Delta t$, $k=0,\dots,N$, and $T=N\Delta t$. 
The control is a sequence $\{u_k\}_{k=0}^{N-1}$ with $u_k\in U$, and 
the state evolves on a smooth manifold $Q$.


We assume discrete dynamics of the form
\begin{equation}
  x_{k+1} = F(x_k,u_k),
  \label{eq:discrete-dynamics}
\end{equation}
where $F:Q\times U\to Q$ is a $C^1$ map and $x_0$ is fixed.  

The discrete Bolza cost functional is
\begin{equation}
  J_d(\{u_k\}) = \sum_{k=0}^{N-1} L_d(x_k,u_k) + \Phi(x_N).
\end{equation}

In contrast to standard discretizations where 
$L_d(x,u)=L(x,u)\Delta t+O(\Delta t^2)$ is chosen ad hoc, we take $L_d$ to be a 
consistent approximation of the \emph{exact discrete contact Lagrangian} 
associated with the continuous contact PMP flow, in the sense of 
\cite{VermeerenContactVI,Anahory2021JNS}.  
This ensures that the resulting discrete scheme preserves the contact geometry.

As in the continuous case, we introduce an accumulated cost variable:
\begin{equation}
  z_{k+1} = z_k + L_d(x_k,u_k), 
  \qquad z_0=0.
  \label{eq:z_update}
\end{equation}
Since $L_d$ approximates $\int_{t_k}^{t_{k+1}} L(x(t),u(t))\,\mathrm{d}t$, the time 
step $\Delta t$ is already incorporated, and no explicit factor appears in 
\eqref{eq:z_update}.  
The extended discrete state is $(x_k,z_k)\in Q\times\mathbb{R}$.


Let $p_{k+1}\in T_{x_{k+1}}^\ast Q$ denote the discrete costate associated
with the constraint $x_{k+1}=F(x_k,u_k)$.  
We introduce the \emph{discrete contact Hamiltonian}
\begin{equation}
  H_d(x_k,p_{k+1},z_k,u_k)
    := \big\langle p_{k+1}, F(x_k,u_k)\big\rangle 
       + L_d(x_k,u_k),
  \label{eq:Hd-def}
\end{equation}
which is a smooth map
\[
  H_d : T^\ast Q \times \mathbb{R} \times U \to \mathbb{R},\qquad
  (x_k,p_{k+1},z_k,u_k)\mapsto H_d(x_k,p_{k+1},z_k,u_k).
\]

We also define the associated type-II \emph{discrete contact generating function}
\begin{equation}
  S_d(x_k,p_{k+1},z_k,u_k)
    := z_k + L_d(x_k,u_k) 
       + \big\langle p_{k+1}, F(x_k,u_k)-x_k\big\rangle .
  \label{eq:Sd-def}
\end{equation}

Note that $H_d$ does not depend on the cost variable $z_k$, so we may occasionally write $H_d(x_k,p_{k+1},u_k)$ for brevity.
Its domain is still $T^\ast Q\times\mathbb{R}\times U$, but the dependence on $z_k$ is trivial.

Furthermore, differentiation with respect to the control variable yields
\begin{equation}
  \partial_u S_d(x_k,p_{k+1},z_k,u_k)
    = \partial_u L_d(x_k,u_k)
      + \big\langle p_{k+1},\partial_u F(x_k,u_k)\big\rangle
    = \partial_u H_d(x_k,p_{k+1},z_k,u_k).
  \label{eq:dSu_equals_dHu}
\end{equation}
Hence the stationarity conditions $\partial_u S_d = 0$ and
$\partial_u H_d = 0$ are equivalent.

The specific combination of terms in $S_d$ is not arbitrary.  
In the continuous contact PMP, the identity 
$\dot z = \langle p,\dot x\rangle - H_c(x,p,z,u)$ 
characterizes contact Hamiltonian flows.  
Up to the sign convention on the cost multiplier, this can be rewritten using
a Hamiltonian of the form $\widetilde H(x,p,u)=\langle p,f(x,u)\rangle + L(x,u)$,
which is the natural form arising from the finite-dimensional Lagrange 
multiplier derivation.  
Our discrete Hamiltonian $H_d$ follows this latter convention.  

The discrete generating function $S_d$ reproduces the same structural pattern:
it contains $z_k$, the discrete analogue of the accumulated cost, and the term 
$\langle p_{k+1},F(x_k,u_k)-x_k\rangle + L_d(x_k,u_k)$ in exactly the same 
additive configuration as in the continuous identity.  
Moreover, a direct computation using \eqref{eq:Hd-def} and \eqref{eq:Sd-def} gives
\begin{equation}
  \partial_x S_d(x_k,p_{k+1},z_k,u_k)
    = \partial_x H_d(x_k,p_{k+1},z_k,u_k) - p_{k+1},
  \label{eq:dSd_dx_relation}
\end{equation}
so that the adjoint recursion can be written equivalently as
\[
   p_k - p_{k+1}
    = \partial_x S_d(x_k,p_{k+1},z_k,u_k).
\]

Next, we define the extended discrete manifold $M_d := T^\ast Q \times \mathbb{R}$ and, for each time index $k$, the discrete contact form
\begin{equation}
  \Theta_k := \mathrm{d}z_k - \langle p_k, \mathrm{d}x_k\rangle,
  \label{eq:Theta-k}
\end{equation}
which is the exact discrete analogue of the continuous contact form 
$\theta = \mathrm{d}z - \langle p,\mathrm{d}x\rangle$.  
Differentiating \eqref{eq:Sd-def} and using \eqref{eq:Theta-k} yields the
discrete contact identity
\begin{equation}
  \mathrm{d}S_d
    = \Theta_k 
      + \langle p_{k+1},\mathrm{d}x_{k+1}\rangle
      - \Theta_{k+1},
  \label{eq:dSd-identity_explicit}
\end{equation}
an identity valid on the extended space with coordinates 
$(x_k,p_k,z_k,x_{k+1},p_{k+1},z_{k+1})$ before imposing the discrete 
dynamics. This is the discrete analogue of the continuous contact identity
$\mathcal{L}_{X_H}\theta = (\partial_z H)\theta$ and will be used below to
prove that the resulting PMP update map is a discrete contactomorphism.


We now state the discrete analogue of the contact PMP.

\begin{theorem}
\label{thm:discrete_contact_pmp}
Let $\{u_k^\ast\}_{k=0}^{N-1}$ be optimal for the discrete problem.  
Then there exists a nontrivial costate sequence 
$\{p_k^\ast\}_{k=1}^N$, $p_k^\ast\in T_{x_k^\ast}^\ast Q$, such that 
for $k=0,\dots,N-1$, the quadruple 
$(x_k^\ast,p_{k+1}^\ast,z_k^\ast,u_k^\ast)$ satisfies:

\begin{enumerate}

\item \textbf{State and cost updates:}
\begin{subequations}
\begin{align}
  x_{k+1}^\ast &= F(x_k^\ast,u_k^\ast), \\
  z_{k+1}^\ast &= z_k^\ast + L_d(x_k^\ast,u_k^\ast).
\end{align}
\end{subequations}

\item \textbf{Adjoint equation:} for $k=1,\dots,N-1$,
\begin{equation}
  p_k^\ast 
    = \partial_x L_d(x_k^\ast,u_k^\ast)
      + \big(\partial_x F(x_k^\ast,u_k^\ast)\big)^\ast p_{k+1}^\ast,
  \label{eq:discrete_adjoint}
\end{equation}
or equivalently
\begin{equation}
  p_k^\ast 
    = \partial_x H_d(x_k^\ast,p_{k+1}^\ast,z_k^\ast,u_k^\ast),
\end{equation}
and, using \eqref{eq:dSd_dx_relation},
\begin{equation}
  p_k^\ast - p_{k+1}^\ast
    = \partial_x S_d(x_k^\ast,p_{k+1}^\ast,z_k^\ast,u_k^\ast).
  \label{eq:discrete_adjoint_Sd}
\end{equation}

\item \textbf{Terminal condition:}
\begin{equation}
  p_N^\ast = \mathrm{d}_x \Phi(x_N^\ast).
\end{equation}

\item \textbf{Stationarity:}
\begin{equation}
  \partial_u H_d(x_k^\ast,p_{k+1}^\ast,z_k^\ast,u_k^\ast) = 0,
  \qquad k=0,\dots,N-1,
  \label{eq:discrete_stationarity_dH}
\end{equation}
or, under the standing continuity and compactness assumptions,
\begin{equation}
  u_k^\ast \in 
  \arg\max_{u\in U} H_d(x_k^\ast,p_{k+1}^\ast,z_k^\ast,u),
  \label{eq:discrete_stationarity}
\end{equation}
for $k=0,\dots,N-1$.

\end{enumerate}
\end{theorem}

\begin{proof}
We use the classical finite-dimensional Lagrange multiplier theorem in local
coordinates on $Q$ (see, e.g., \cite{roca}); see also 
\cite{AssifChatterjeeBanavar2020} for a related argument in a purely geometric setting.

\medskip

Fix $x_0\in Q$, a time step $\Delta t>0$ and an integer $N\in\mathbb{N}$ defining the horizon $T=N\Delta t$. An admissible discrete trajectory is completely
determined by the sequence of states $x_1,\dots,x_N$ and controls 
$u_0,\dots,u_{N-1}$ satisfying the discrete dynamics
\begin{equation}
  x_{k+1} = F(x_k,u_k), 
  \qquad k=0,\dots,N-1.
  \label{eq:constraint-dynamics}
\end{equation}
The discrete cost is
\[
  J_d(\{x_k,u_k\})
    = \sum_{k=0}^{N-1} L_d(x_k,u_k) + \Phi(x_N),
\]
and the accumulated cost variable $z_k$ is defined by
\[
  z_{k+1} = z_k + L_d(x_k,u_k), \qquad z_0=0,
\]
so it does not introduce additional constraints.

We assume throughout that $U$ is a compact subset of a finite-dimensional
vector space, that $F$ and $L_d$ are $C^1$, and that $\Phi$ is $C^1$.
Under these assumptions, the discrete optimal control problem reduces to a
finite-dimensional constrained optimization problem with decision variables 
$(x_1,\dots,x_N,u_0,\dots,u_{N-1})$, equality constraints given by 
\eqref{eq:constraint-dynamics}, and cost functional $J_d$ as above.

We focus on a \emph{normal} optimal solution, i.e.\ one to which the usual
Lagrange multiplier theorem applies with nonzero cost multiplier; this
excludes pathological abnormal extremals and is standard in PMP theory.

\medskip

Let $\{x_k^\ast\}_{k=0}^N$ and $\{u_k^\ast\}_{k=0}^{N-1}$ be an optimal
trajectory and control sequence.
Since $Q$ is a smooth manifold and the time horizon is finite, we may choose
coordinate charts $\varphi_k:U_k\to\mathbb{R}^n$ with $x_k^\ast\in U_k$
for $k=0,\dots,N$, such that the image of $x_k^\ast$ lies in the interior of
$\varphi_k(U_k)$.
Working in these local coordinates, we identify $\xi_k := \varphi_k(x_k) \in \mathbb{R}^n$, and rewrite 
$F(x_k,u_k)$ as a smooth map 
\(
  \widehat{F}_k(\xi_k,u_k)
\)
taking values in $\mathbb{R}^n$.
Similarly, $L_d$ and $\Phi$ are expressed in local coordinates as smooth maps
on $\mathbb{R}^n\times U$ and $\mathbb{R}^n$, respectively.
Since the Lagrange multiplier conditions are coordinate-invariant, it suffices
to derive them in this local representation.

In these coordinates, the discrete optimization problem becomes
\[
  \min \;\; 
  \widehat{J}_d(\{\xi_k,u_k\})
  := 
  \sum_{k=0}^{N-1}\widehat{L}_d(\xi_k,u_k)
  + \widehat{\Phi}(\xi_N),
\]
subject to
\begin{equation}
  \xi_{k+1} = \widehat{F}_k(\xi_k,u_k),
  \qquad k=0,\dots,N-1,
  \label{eq:coord-constraints}
\end{equation}
with $\xi_0$ fixed.

\medskip

We now apply the classical finite-dimensional Lagrange multiplier theorem
for equality constraints. 
Define the constraint map
\[
  G_k(\xi_k,u_k,\xi_{k+1})
    := \widehat{F}_k(\xi_k,u_k) - \xi_{k+1} \in \mathbb{R}^n,
  \qquad k=0,\dots,N-1.
\]
The full constraint map is
\[
  G(\{\xi_k,u_k\})
    := (G_0,\dots,G_{N-1}) \in (\mathbb{R}^n)^N.
\]
Since $G$ is $C^1$ and the feasible set is nonempty, the usual Lagrange
multiplier theorem guarantees the existence of multipliers
$\{\lambda_{k+1}\}_{k=0}^{N-1}\subset(\mathbb{R}^n)^\ast$ and a scalar
$\alpha\in\mathbb{R}$ (cost multiplier) such that, at the optimum,
\[
  \alpha\,\mathrm{D}\widehat{J}_d
  + \sum_{k=0}^{N-1} \lambda_{k+1}^\top \mathrm{D}G_k = 0.
\]
Normality means $\alpha\neq 0$, and we normalize by setting $\alpha=1$.
The multipliers $\lambda_{k+1}$ will give, after identification via the charts,
the covectors $p_{k+1}^\ast\in T_{x_{k+1}^\ast}^\ast Q$.

Equivalently, introducing the augmented functional
\[
  \widehat{\mathbb{J}}_d
   =
   \sum_{k=0}^{N-1}
   \Big(
     \widehat{L}_d(\xi_k,u_k)
     + \lambda_{k+1}^\top\big(\widehat{F}_k(\xi_k,u_k)-\xi_{k+1}\big)
   \Big)
   + \widehat{\Phi}(\xi_N),
\]
first-order optimality reads
\(
  \partial_{\xi_k}\widehat{\mathbb{J}}_d = 0
\)
for $k=1,\dots,N$, and
\(
  \partial_{u_k}\widehat{\mathbb{J}}_d=0
\)
for $k=0,\dots,N-1$.

\medskip

For each $k$, differentiation of $\widehat{\mathbb{J}}_d$ with respect to
$\lambda_{k+1}$ yields
\[
  0 = \partial_{\lambda_{k+1}}\widehat{\mathbb{J}}_d
    = \widehat{F}_k(\xi_k,u_k)-\xi_{k+1},
\]
i.e.\ the constraint \eqref{eq:coord-constraints}.
Translating back to $Q$, this is precisely $x_{k+1} = F(x_k,u_k)$, which proves the state equation in the theorem.
The cost update $z_{k+1}=z_k+L_d(x_k,u_k)$ holds by definition of $z_k$.

\medskip

For $k=1,\dots,N-1$, the derivative of $\widehat{\mathbb{J}}_d$ with respect
to $\xi_k$ gives
\[
  0 = \partial_{\xi_k}\widehat{\mathbb{J}}_d
    = \partial_{\xi_k}\widehat{L}_d(\xi_k,u_k)
      + \big(\partial_{\xi_k}\widehat{F}_k(\xi_k,u_k)\big)^\top\lambda_{k+1}
      - \lambda_k.
\]
Rearranging,
\[
  \lambda_k
    = \partial_{\xi_k}\widehat{L}_d(\xi_k,u_k)
      + \big(\partial_{\xi_k}\widehat{F}_k(\xi_k,u_k)\big)^\top\lambda_{k+1}.
\]
Pulling this back via the charts yields, in intrinsic notation,
\[
  p_k
    = \partial_x L_d(x_k,u_k)
      + (\partial_x F(x_k,u_k))^\ast p_{k+1},
\]
where $p_k\in T_{x_k}^\ast Q$ is the covector corresponding to $\lambda_k$.
This is exactly the adjoint recursion \eqref{eq:discrete_adjoint}.

Using the definition \eqref{eq:Hd-def} of $H_d$, one computes
\[
  \partial_x H_d(x_k,p_{k+1},z_k,u_k)
    = \partial_x\big\langle p_{k+1},F(x_k,u_k)\big\rangle
      + \partial_x L_d(x_k,u_k)
    = (\partial_x F(x_k,u_k))^\ast p_{k+1}
      + \partial_x L_d(x_k,u_k),
\]
so the previous relation is equivalently written as
\[
  p_k
    = \partial_x H_d(x_k,p_{k+1},z_k,u_k),
\]
and, by \eqref{eq:dSd_dx_relation}, as
\[
  p_k - p_{k+1}
    = \partial_x S_d(x_k,p_{k+1},z_k,u_k),
\]
which gives \eqref{eq:discrete_adjoint_Sd}.

\medskip

Similarly, differentiation of $\widehat{\mathbb{J}}_d$ with respect to $\xi_N$
yields
\[
  0 = \partial_{\xi_N}\widehat{\mathbb{J}}_d
    = \partial_{\xi_N}\widehat{\Phi}(\xi_N) - \lambda_N,
\]
so $\lambda_N = \partial_{\xi_N}\widehat{\Phi}(\xi_N)$, i.e. 
$p_N = \mathrm{d}_x\Phi(x_N)$, which gives the terminal condition in the theorem.

\medskip

For each $k=0,\dots,N-1$, the derivative with respect to $u_k$ is
\[
  0 = \partial_{u_k}\widehat{\mathbb{J}}_d
    = \partial_{u_k}\widehat{L}_d(\xi_k,u_k)
      + \big(\partial_{u_k}\widehat{F}_k(\xi_k,u_k)\big)^\top\lambda_{k+1}.
\]
In intrinsic notation this reads
\[
  0
  = \partial_u L_d(x_k,u_k)
    + (\partial_u F(x_k,u_k))^\ast p_{k+1}.
\]
Using the definition \eqref{eq:Hd-def} of $H_d$,
\[
  \partial_u H_d(x_k,p_{k+1},z_k,u_k)
    = (\partial_u F(x_k,u_k))^\ast p_{k+1}
      + \partial_u L_d(x_k,u_k),
\]
so the first-order condition above is equivalent to
\[
  \partial_u H_d(x_k,p_{k+1},z_k,u_k) = 0,
\]
which is \eqref{eq:discrete_stationarity_dH}.  
Under the standing assumptions (continuity of $H_d$ in $u$, compactness of $U$)
this first-order condition is equivalent to
\[
  u_k \in \arg\max_{u\in U} H_d(x_k,p_{k+1},z_k,u),
\]
which is the maximum condition \eqref{eq:discrete_stationarity}.

\medskip

The Lagrange multiplier theorem implies that the pair consisting of the cost
multiplier $\alpha$ and the constraint multipliers $\{\lambda_k\}$ cannot be
trivial. 
Under the normality assumption, we have $\alpha\neq 0$, which we normalized
to $\alpha=1$. 
Therefore the sequence $\{\lambda_k\}$, and hence $\{p_k\}$, cannot vanish
identically. 
This yields a nontrivial costate sequence $\{p_k^\ast\}_{k=1}^N$ associated
with the optimal control $\{u_k^\ast\}_{k=0}^{N-1}$.

\medskip

Collecting all the conditions above and reverting to the original notation,
we obtain exactly the system stated in 
Theorem~\ref{thm:discrete_contact_pmp}.
\end{proof}

\begin{remark}
If the Lagrange multiplier associated with the cost functional vanishes, 
i.e.\ $\alpha=0$ in the finite-dimensional multiplier theorem, the discrete 
PMP reduces to the system obtained by removing all terms involving $L_d$ and 
$\Phi$.  
The adjoint recursion becomes 
\[
  p_k = (\partial_x F(x_k,u_k))^\ast p_{k+1},
\]
and the stationarity condition reduces to 
\[
  (\partial_u F(x_k,u_k))^\ast p_{k+1}=0.
\]
Such extremals do not see the cost and depend only on the geometry of the 
constraint $x_{k+1}=F(x_k,u_k)$.  
In many practical situations they are excluded automatically because, for 
each $k$, the map $u \mapsto F(x_k,u)$ has full rank in the control directions, 
forcing $p_{k+1}=0$ and hence triviality of the entire multiplier sequence. 
\hfill$\diamond$
\end{remark}

\begin{remark} 
The regularity assumption $F\in C^1(Q\times U,Q)$ is the minimal smoothness 
required for the discrete PMP.  
Indeed, the discrete Hamiltonian $H_d$ and the contact generating function 
$S_d$ involve derivatives with respect to $x_k$ and $u_k$, and the adjoint 
update \eqref{eq:discrete_adjoint} contains the term $\partial_xF(x_k,u_k)$.  
Thus differentiability of $F$ is necessary and sufficient for the variational 
derivation of the discrete PMP, for the definition of the discrete costate 
recursion, and for the validity of the stationarity condition.  
Higher regularity (e.g.\ $C^2$) is only required for error analysis, not for 
the first-order optimality theory or the discrete contactomorphism property. \hfill$\diamond$
\end{remark}


The discrete PMP update map is $\Psi_k : (x_k,p_{k+1},z_k) \to (x_{k+1},p_k,z_{k+1})$, where the triple $(x_{k+1},p_k,z_{k+1})$ is defined implicitly by
\begin{subequations}
\label{eq:Psi_k_def}
\begin{align}
  x_{k+1} &= F(x_k,u_k), \label{eq:Psi_state}\\
  z_{k+1} &= z_k + L_d(x_k,u_k), \label{eq:Psi_cost}\\
  p_k - p_{k+1} &= \partial_{x_k} S_d(x_k,p_{k+1},z_k,u_k), \label{eq:Psi_adj}\\
  0 &= \partial_{u_k} S_d(x_k,p_{k+1},z_k,u_k), \label{eq:Psi_stat}
\end{align}
\end{subequations}
and $u_k\in U$ is determined by the stationarity condition
$\partial_{u_k} S_d(x_k,p_{k+1},z_k,u_k)=0$.
Using $\partial_u S_d=\partial_u H_d$ from
\eqref{eq:dSu_equals_dHu}, this is equivalent to the
discrete condition $\partial_u H_d=0$ and, under the standing
compactness assumptions, to the maximization condition
\eqref{eq:discrete_stationarity}.
 
Note that \eqref{eq:Psi_state} and \eqref{eq:Psi_adj} are consistent with the 
discrete PMP equations in Theorem~\ref{thm:discrete_contact_pmp}, since
\[
  \partial_{p_{k+1}} S_d(x_k,p_{k+1},z_k,u_k)
    = F(x_k,u_k) - x_k,
\]
so that \eqref{eq:Psi_state} can be rewritten as
\[
  x_{k+1} - x_k = \partial_{p_{k+1}} S_d(x_k,p_{k+1},z_k,u_k),
\]
and \eqref{eq:Psi_adj} is precisely the equivalent form
\eqref{eq:discrete_adjoint_Sd} of the adjoint recursion.

Only the control $u_k$ is obtained from a maximization of $S_d$ (or $H_d$), 
while $x_{k+1}$, $p_k$, and $z_{k+1}$ follow from the discrete dynamics and 
the adjoint/contact equations generated by $S_d$.

\begin{remark}
The discrete adjoint equation involves $p_{k+1}$ rather than $p_{k}$ because the 
discrete constraint $x_{k+1}=F(x_k,u_k)$ introduces its Lagrange multiplier at the 
future node. This mirrors the continuous PMP, in which $p(t)$ satisfies a 
terminal boundary condition and evolves backward in time. In the language of 
discrete contact and Hamilton--Jacobi theory, the function 
$S_d(x_k,p_{k+1},z_k,u_k)$ plays the role of a type-II generating function, 
which naturally defines the update map 
$(x_k,p_{k+1})\mapsto (x_{k+1},p_k)$ after eliminating $u_k$ and $z_k$.
\hfill$\diamond$
\end{remark}

\begin{theorem}
\label{thm:contactomorphism}
For each $k$, the discrete PMP map $\Psi_k$ satisfies
\begin{equation}
  \Psi_k^\ast \Theta_{k+1}
  = \Theta_k.
  \label{eq:discrete-contactomorphism}
\end{equation}
Thus $\Psi_k$ is a strict discrete contactomorphism, preserving the discrete
family of contact forms $\{\Theta_k\}$ on $M_d$.
\end{theorem}

\begin{proof}
Differentiating \eqref{eq:Sd-def} and using the definition of $\Theta_k$ 
in \eqref{eq:Theta-k} yields the discrete contact identity
\eqref{eq:dSd-identity_explicit},
\[
  \mathrm{d}S_d
    = \Theta_k 
      + \langle p_{k+1},\mathrm{d}x_{k+1}\rangle
      - \Theta_{k+1},
\]
an identity valid on the extended space with coordinates 
$(x_k,p_k,z_k,x_{k+1},p_{k+1},z_{k+1})$ before imposing the discrete 
dynamics.

On the graph of the update map $\Psi_k$, the variables 
$(x_{k+1},p_k,z_{k+1})$ are expressed in terms of $(x_k,p_{k+1},z_k)$ by
\eqref{eq:Psi_k_def}, and the control $u_k$ is determined by the 
stationarity condition \eqref{eq:Psi_stat}.  
Restricting \eqref{eq:dSd-identity_explicit} to this graph and using
\eqref{eq:Psi_state}--\eqref{eq:Psi_cost} to eliminate 
$\mathrm{d}x_{k+1}$ and $\mathrm{d}z_{k+1}$ yields
\[
  \mathrm{d}S_d
    = \Theta_k - \Psi_k^\ast \Theta_{k+1},
\]
where now $S_d$ is viewed as a function of the independent variables
$(x_k,p_{k+1},z_k)$.

By construction, $S_d$ is a type-II contact generating function for the map
$\Psi_k$; the stationarity condition \eqref{eq:Psi_stat} eliminates the
dependence on the control, so $\mathrm{d}S_d$ is an exact 1-form on the
parameter space of the graph.  
Hence the identity above implies $\Psi_k^\ast \Theta_{k+1}
  = \Theta_k$, which is \eqref{eq:discrete-contactomorphism}.

Since $H_d$ (and therefore $S_d$) is independent of the cost variable $z_k$,
the associated discrete contact vector field has vanishing $\partial_z H_d$.
Consequently, the contactomorphism $\Psi_k$ preserves the discrete contact
form exactly, with no conformal factor.  
This is the discrete analogue of the strict contact case for the continuous
PMP discussed in Section~\ref{sec:contact_discrete_pmp}.\end{proof}

\begin{remark}
Assif--Chatterjee--Banavar~\cite{AssifChatterjeeBanavar2020} obtained a discrete 
PMP on manifolds, but without identifying or enforcing preservation of the 
contact structure.  
Theorem~\ref{thm:contactomorphism} shows that our discrete PMP map 
preserves the discrete contact form exactly, matching the continuous property
$\mathcal{L}_{X_H}\theta = (\partial_z H)\theta$ and providing a 
structure-preserving discretization of the contact PMP.
\hfill$\diamond$
\end{remark}

\begin{remark}
The discrete contact structure is intrinsically tied to the \emph{normal} case.  
When the cost multiplier $\alpha\neq 0$, the discrete PMP is generated by the 
type-II contact generating function $S_d$, and the update map $\Psi_k$ is a 
strict discrete contactomorphism (Theorem~\ref{thm:contactomorphism}).  

In contrast, for abnormal extremals ($\alpha=0$) the cost functional plays no 
role and all terms involving $L_d$ and $\Phi$ disappear.  
The resulting dynamics satisfy only
\[
p_k = (\partial_x F(x_k,u_k))^\ast p_{k+1},
\qquad
(\partial_u F(x_k,u_k))^\ast p_{k+1}=0,
\]
which define a \emph{presymplectic} rather than contact update map.  
This mirrors precisely the continuous-time situation studied by 
de~Le\'on--Lainz--Mu\~noz-Lecanda, where abnormal extremals live on a 
presymplectic PMP manifold and cannot be captured by a contact reduction.  
Thus the discrete theory reproduces the same geometric dichotomy: 
normal extremals are contact, abnormal ones presymplectic.
\hfill$\diamond$
\end{remark}

\section{Contact Lie--Group Variational Integrators}
\label{sec:contact_LGVI}

In this section we adapt the contact--geometric framework of 
Section~\ref{sec:discrete_contact_pmp} to the case where the configuration 
manifold is a Lie group $G$.  
We work in left--trivialized coordinates and construct retraction--based 
Lie--group contact variational integrators.  
These integrators will later be combined with the discrete contact PMP to 
obtain structure--preserving schemes for optimal control on $G$, and in 
particular on $\SU(2)$.
Throughout this section, all trivializations of $TG$ and $T^\ast G$ are understood to be left--trivializations unless explicitly stated otherwise. 

Throughout, $\fg$ denotes the Lie algebra of $G$.  
We denote by
\[
    \mathrm{L}_g : G \to G,\qquad \mathrm{L}_g(h)=gh,
    \qquad
    \mathrm{R}_g : G \to G,\qquad \mathrm{R}_g(h)=hg,
\]
the left and right translations by $g\in G$.  
Their tangent maps at the identity are
\[
    T_e\mathrm{L}_g : \fg \to T_g G,
    \qquad
    T_e\mathrm{R}_g : \fg \to T_g G,
\]
and the corresponding dual maps (used to left-- and right--trivialize cotangent vectors) are
\[
    T_e^\ast\mathrm{L}_g : T_g^\ast G \to \fg^\ast,
    \qquad
    T_e^\ast\mathrm{R}_g : T_g^\ast G \to \fg^\ast.
\]

The adjoint representation of $G$ on its Lie algebra $\fg$ is defined by
\[
    \Ad_g 
      := T_e\!\left( \mathrm{L}_g \circ \mathrm{R}_{g^{-1}} \right)
      : \fg \to \fg,
\]
and its dual action by $\Ad_g^\ast : \fg^\ast \to \fg^\ast$.

\subsection{Contact Lagrangian systems on Lie groups}
\label{subsec:contact_LG_continuous}

Let $G$ be a Lie group with Lie algebra $\fg$, and let 
$L : TG \times \R \to \R$ be a (regular) contact Lagrangian in the sense of \cite{herglotzLG}.  
Using the standard left trivialization
\[
  TG \times \R \;\simeq\; G \times \fg \times \R, 
  \qquad 
  (g,\dot g,z) \mapsto \bigl(g,\xi,z\bigr),\quad 
  \xi := T_g \mathrm{L}_{g^{-1}}(\dot g),
\]
we obtain the \emph{left--trivialized contact Lagrangian}
\[
  \widetilde L : G \times \fg \times \R \to \R,
  \qquad
  \widetilde L(g,\xi,z) := L(g, \dot g, z),
  \quad \dot g = T_e \mathrm{L}_g(\xi).
\]
For simplicity of notation, we will write $L(g,\xi,z)$ instead of $\widetilde L$ 
from now on.

\medskip

The dynamics of the contact Lagrangian system on $TG\times\R$ can be expressed,
in these left--trivialized coordinates, by the following 
\emph{Euler--Poincar\'e--Herglotz} equations on $G\times\fg\times\R$ 
(cf.\ Theorem~4.1 in \cite{herglotzLG}):
\begin{subequations}
\label{eq:EPH-Lie-group}
\begin{align}
  \dot g &= T_e \mathrm{L}_g(\xi),
  \label{eq:EPH-g}\\[1mm]
  \dot z &= L(g,\xi,z),
  \label{eq:EPH-z}\\[1mm]
  \frac{d}{dt}\left(\frac{\delta L}{\delta \xi}(g,\xi,z)\right)
    &= \ad^*_{\xi}\!\left(\frac{\delta L}{\delta \xi}(g,\xi,z)\right)
       + T_e^\ast \mathrm{L}_g\!\left(\frac{\delta L}{\delta g}(g,\xi,z)\right)
       + \frac{\delta L}{\delta \xi}(g,\xi,z)\,\frac{\partial L}{\partial z}(g,\xi,z).
  \label{eq:EPH-momentum}
\end{align}
\end{subequations}
Here 
\(
  \delta L / \delta \xi \in \fg^\ast
\) 
is the functional derivative with respect to the Lie--algebra variable,
\(
  \delta L / \delta g \in T_g^\ast G
\)
is the derivative with respect to $g$, 
$\ad^*_\xi : \fg^\ast \to \fg^\ast$ is the coadjoint action,
and 
\(
  T_e^\ast\mathrm{L}_g : T_g^\ast G \to \fg^\ast
\)
is as above.

\medskip

Defining the \emph{body momentum} $\displaystyle{\mu := \frac{\delta L}{\delta \xi}(g,\xi,z) \in \fg^\ast}$, equation~\eqref{eq:EPH-momentum} can be written more compactly as
\begin{equation}
  \dot\mu
    = \ad^*_\xi \mu 
      + T_e^\ast \mathrm{L}_g\!\Bigl(\frac{\delta L}{\delta g}\Bigr)
      + \mu\,\frac{\partial L}{\partial z},
  \qquad
  \xi = T_g \mathrm{L}_{g^{-1}}(\dot g),
  \qquad
  \dot z = L(g,\xi,z).
  \label{eq:EPH-compact}
\end{equation}

\medskip

Assuming $L$ is \emph{hyperregular} (so that the fiber derivative
\(
  \mathbb{F}L : TG\times\R \to T^\ast G\times\R
\)
is a global diffeomorphism), one defines the left--trivial Legendre transform
\[
  \mathbb{F}L : G\times\fg\times\R \to G\times\fg^\ast\times\R,
  \qquad
  \mathbb{F}L(g,\xi,z) = (g,\mu,z),
  \quad \mu = \frac{\delta L}{\delta \xi}(g,\xi,z),
\]
and the associated \emph{contact Hamiltonian}
\[
  H : G\times\fg^\ast\times\R \to \R,
  \qquad
  H(g,\mu,z) := \langle \mu,\xi\rangle - L(g,\xi,z),
  \quad \mu = \frac{\delta L}{\delta \xi}(g,\xi,z).
\]
The Legendre transform is a contactomorphism between 
$(TG\times\R,\eta_L)$ and $(T^\ast G\times\R,\theta)$, where
\[
  \eta_L = \mathrm{d}z - \theta_L,
  \qquad
  \theta_L = (\mathbb{F}L)^\ast\bigl(\theta_G\bigr),
  \qquad
  \theta_G = \langle p, g^{-1}\mathrm{d}g\rangle
\]
is the left--trivialized canonical $1$--form on $T^\ast G$, and
\[
  \theta = \mathrm{d}z - \theta_G
\]
is the canonical contact form on $T^\ast G\times\R$.
In particular, the vector field $\xi_L$ solving the 
Euler--Poincar\'e--Herglotz equations \eqref{eq:EPH-Lie-group} is 
$\mathbb{F}L$--related to the contact Hamiltonian vector field $X_H$ on 
$(T^\ast G\times\R,\theta)$:
\[
  (\mathbb{F}L)_\ast(\xi_L) = X_H \circ \mathbb{F}L,
\]
and the resulting flow satisfies the identity
\[
  \mathcal{L}_{X_H}\theta = \bigl(\partial_z H\bigr)\,\theta.
\]

\begin{remark}
In the special case of a left--invariant contact Lagrangian, 
$L(g,\xi,z) = \ell(\xi,z)$, the term 
$T_e^\ast \mathrm{L}_g(\delta L/\delta g)$ vanishes and 
\eqref{eq:EPH-compact} reduces to the \emph{Euler--Poincar\'e--Herglotz}
equations on $\fg\times\R$:
\[
  \dot\mu = \ad^*_\xi\mu + \mu\,\frac{\partial \ell}{\partial z}(\xi,z),
  \qquad
  \dot z = \ell(\xi,z),
  \qquad
  \mu = \frac{\partial \ell}{\partial \xi}(\xi,z).
\]
This is precisely the reduced form obtained in 
\cite{herglotzLG}, and it will be the relevant setting for 
our Lie--group variational integrators.
\hfill$\diamond$
\end{remark}

\subsection{Exact discrete contact Lagrangian and order of a contact LGVI}
\label{subsec:exact_contact_Ld_Lie}

Let $h>0$ be a fixed time step.  
Let $L : G \times \fg \times \R \to \R$ be a smooth left--trivialized contact
Lagrangian, and let $(g(t),\xi(t),z(t))$ with $t\in[0,h]$ be the unique solution of the Euler--Poincar\'e--Herglotz equations
\eqref{eq:EPH-Lie-group} with boundary conditions $g(0)=g_k$, $z(0)=z_k$ and $g(h)=g_{k+1}$ for $h>0$ small enough. Existence and uniqueness hold under the usual regularity and hyperregularity
assumptions on $L$.

\medskip

\noindent
The \emph{exact discrete contact Lagrangian} is then defined by
\begin{equation}
  L_h^{\mathrm{e}}(g_k,g_{k+1},z_k)
    :=
    \int_0^h 
      L\bigl(g(t),\xi(t),z(t)\bigr)\,\mathrm{d}t,
  \label{eq:exact_contact_Ld}
\end{equation}
where $(g(t),\xi(t),z(t))$ is as above.

\begin{proposition}
\label{prop:exact-contact-Lie}
Let $L_h^{\mathrm{e}}$ be given by \eqref{eq:exact_contact_Ld}.  
Define the discrete action sum
\[
  \mathcal{S}_d 
    = \sum_{k=0}^{N-1} L_h^{\mathrm{e}}(g_k,g_{k+1},z_k),
\]
with $z_k$ determined recursively by
\begin{equation}
  z_{k+1}
    = z_k + L_h^{\mathrm{e}}(g_k,g_{k+1},z_k),
  \qquad z_0 \text{ given.}
  \label{eq:z-recursion-exact}
\end{equation}

Then the discrete Herglotz principle $\delta \mathcal{S}_d = 0$ for variations $\delta g_k\in T_{g_k}G$ vanishing at endpoints yields a discrete evolution whose image under the contact Legendre transform $\mathbb{F}L:TG\times\mathbb{R}\to T^{*}G\times\mathbb{R}$ coincides with the time–$h$ contact Hamiltonian flow of the continuous system associated with $L$

In particular, if $\Phi_H^h : T^\ast G\times\R \to T^\ast G\times\R$ denotes
the time--$h$ flow of the contact Hamiltonian vector field $X_H$, and
\(
  \theta = \mathrm{d}z - \langle p, g^{-1}\mathrm{d}g\rangle
\)
is the canonical contact form on $T^\ast G\times\R$, then
\begin{equation}
  (\Phi_H^h)^\ast \theta
    =
    \exp\!\left(
      \int_0^h \bigl(\partial_z H\bigr)\!\circ\Phi_H^t\,\mathrm{d}t
    \right)\theta.
  \label{eq:exact-conformal-factor}
\end{equation}
Equivalently, in first--order form at $(g_k,p_k,z_k)$,
\begin{equation}
  (\Phi_H^h)^\ast \theta
    =
    \bigl(1 + h\,\partial_z H(g_k,p_k,z_k) + O(h^2)\bigr)\,\theta.
  \label{eq:conformal-expansion}
\end{equation}
\end{proposition}

\begin{proof}
We split the argument into two parts: first we show that the discrete
Herglotz principle with $L_h^{\mathrm e}$ reproduces the sampled
Euler--Poincar\'e--Herglotz flow, and then we derive the conformal
factor for the contact form.

\medskip

\emph{Step 1: Exact discrete Lagrangian and sampled flow.}
Fix $h>0$ and consider a solution $(g(t),\xi(t),z(t)),\, t\in[0,h]$ of the Euler--Poincar\'e--Herglotz equations \eqref{eq:EPH-Lie-group} with
boundary data $g(0)=g_k$, $z(0)=z_k$, $g(h)=g_{k+1}$.

By definition,
\[
  L_h^{\mathrm e}(g_k,g_{k+1},z_k)
    = \int_0^h L\bigl(g(t),\xi(t),z(t)\bigr)\,\mathrm{d}t.
\]
Along the same solution we have $\dot z(t)=L(g(t),\xi(t),z(t))$, hence
\[
  z(h)-z(0)
    = \int_0^h L\bigl(g(t),\xi(t),z(t)\bigr)\,\mathrm{d}t
    = L_h^{\mathrm e}(g_k,g_{k+1},z_k).
\]
Identifying $z(0)=z_k$ and $z(h)=z_{k+1}$ gives exactly the recursion
\eqref{eq:z-recursion-exact},
\[
  z_{k+1}
    = z_k + L_h^{\mathrm e}(g_k,g_{k+1},z_k).
\]

Now consider a partition $t_k=kh$, $k=0,\dots,N$, and a discrete path
$\{g_k\}_{k=0}^N$ with fixed endpoints $g_0,g_N$ and initial cost
$z_0$. For each $k$, let $(g_k^{\mathrm{cont}}(t),\xi_k^{\mathrm{cont}}(t),
   z_k^{\mathrm{cont}}(t))$, $t\in[0,h]$, be the unique solution of \eqref{eq:EPH-Lie-group} with $g_k^{\mathrm{cont}}(0)=g_k$, $z_k^{\mathrm{cont}}(0)=z_k$, and $g_k^{\mathrm{cont}}(h)=g_{k+1}$, and define $z_{k+1}$ by \eqref{eq:z-recursion-exact}. Concatenating
these segments yields a piecewise smooth curve $(g(t),\xi(t),z(t)),\, t\in[0,Nh]$, with $g(t_k)=g_k$ and $z(t_k)=z_k$. By uniqueness of solutions of \eqref{eq:EPH-Lie-group}, the resulting curve is in fact the restriction to $[0,Nh]$ of a global solution of the Euler--Poincar\'e--Herglotz equations with the given initial data.

Summing the relations $z_{k+1}-z_k = L_h^{\mathrm e}(g_k,g_{k+1},z_k)$
we obtain
\[
  z_N - z_0
    = \sum_{k=0}^{N-1} L_h^{\mathrm e}(g_k,g_{k+1},z_k)
    = \mathcal{S}_d.
\]
Thus the discrete action $\mathcal{S}_d$ coincides with the total
change in the cost accumulator:
\[
  \mathcal{S}_d = z_N - z_0.
\]
In particular, the discrete Herglotz principle
$\delta\mathcal{S}_d=0$ is equivalent to the condition $\delta z_N = 0$ under variations $\delta g_k\in T_{g_k}G$ with $\delta g_0=\delta g_N=0$,
since $z_0$ is fixed.

On the other hand, the continuous Herglotz variational principle for
the contact Lagrangian $L$ (in Euler--Poincar\'e--Herglotz form)
states that $\delta z(T)=0$ for all variations $\delta g(t)$ with
$\delta g(0)=\delta g(T)=0$ is equivalent to the system
\eqref{eq:EPH-Lie-group} on $[0,T]$. Applying this with $T=Nh$ and
restricting to variations compatible with the discretization (i.e.,
encoded by variations of the nodes $\{g_k\}$) shows that
\[
  \delta z_N = 0
  \quad\Longleftrightarrow\quad
  (g(t),\xi(t),z(t)) \ \text{solves \eqref{eq:EPH-Lie-group} on }[0,Nh].
\]

Therefore, critical discrete curves for $\mathcal{S}_d$ are precisely
the samples $\{(g(t_k),z(t_k))\}$ of solutions of the continuous
Euler--Poincar\'e--Herglotz equations at times $t_k=kh$. Passing to
momenta via the contact Legendre transform $\mathbb{F}L : G\times\fg\times\R \to G\times\fg^\ast\times\R$, and then to cotangent variables $(g,p,z)\in T^\ast G\times\R$ via
left trivialization, we obtain that the induced discrete map on
$T^\ast G\times\R$ is exactly the time--$h$ contact Hamiltonian flow
$\Phi_H^h$ of the Hamiltonian system associated with $L$.

\medskip

\emph{Step 2: Conformal factor for the contact form.}
On the Hamiltonian side, the dynamics are governed by the contact
Hamiltonian vector field $X_H$ on $(T^\ast G\times\R,\theta)$, which
satisfies
\[
  \mathcal{L}_{X_H}\theta = (\partial_z H)\,\theta.
\]
Let $\Phi_H^t$ denote its flow. Differentiating the pullback of $\theta$
along the flow and using Cartan's formula gives
\[
  \frac{\mathrm{d}}{\mathrm{d}t}
     \bigl((\Phi_H^t)^\ast\theta\bigr)
    = (\Phi_H^t)^\ast\bigl(\mathcal{L}_{X_H}\theta\bigr)
    = (\Phi_H^t)^\ast\bigl((\partial_z H)\,\theta\bigr)
    = \bigl((\partial_z H)\circ\Phi_H^t\bigr)\,(\Phi_H^t)^\ast\theta.
\]
Thus, for each initial condition, $(\Phi_H^t)^\ast\theta$ solves a scalar
linear ODE of the form
\[
  \frac{\mathrm{d}}{\mathrm{d}t} \alpha(t)
    = a(t)\,\alpha(t),
  \qquad
  \alpha(0)=\theta,
\]
with $a(t) = (\partial_z H)\circ\Phi_H^t$. The solution is
\[
  (\Phi_H^t)^\ast\theta
    = \exp\!\left(\int_0^t a(s)\,\mathrm{d}s\right)\theta
    = \exp\!\left(
        \int_0^t \bigl(\partial_z H\bigr)\!\circ\Phi_H^s\,\mathrm{d}s
      \right)\theta.
\]
Evaluating at $t=h$ yields \eqref{eq:exact-conformal-factor}.

Finally, expanding the exponential for small $h$ around $t=0$ gives
\[
  \exp\!\left(
    \int_0^h \bigl(\partial_z H\bigr)\!\circ\Phi_H^t\,\mathrm{d}t
  \right)
  = 1 + h\,\partial_z H(g_k,p_k,z_k) + O(h^2),
\]
since $\Phi_H^0=\mathrm{id}$ and therefore
$(\partial_z H)\circ\Phi_H^0 = \partial_z H(g_k,p_k,z_k)$. This proves
\eqref{eq:conformal-expansion}.
\end{proof}

\begin{remark}
The exact discrete contact Lagrangian $L_h^{\mathrm{e}}$ depends on the
solution of a boundary--value problem for the Euler--Poincar\'e--Herglotz
equations and is therefore not computable in closed form except in very
special situations.  
Its main role is conceptual: it provides the ``gold standard'' against which
practical approximations $L_d$ can be designed and compared, ensuring that
the resulting integrators are \emph{contact--structure consistent}.  
This is analogous to the role of exact discrete Lagrangians in
symplectic Lie--group variational integrators 
\cite{MarsdenWest2001,LeokShingel2012,marrero}.
\hfill$\diamond$
\end{remark}

\begin{definition}[Order of a discrete contact integrator]
\label{def:order-contact-integrator}
Let $(g(t),p(t),z(t))$ denote the exact solution of the continuous contact Hamiltonian flow associated with $L$ (equivalently, of the Euler–Poincaré–Herglotz equations in momentum variables), and let
$(g_k,p_k,z_k)$ be the output of a discrete contact variational integrator
(with time step $h$) constructed from a discrete Lagrangian $L_d$.
We say that the method has \emph{order $r$} if, for fixed final time
$T = Nh$,
\[
    \|(g_N,p_N,z_N) - (g(T),p(T),z(T))\|
    = \mathcal{O}(h^r)
    \qquad (h\to 0),
\]
where the norm is taken in any smooth local trivialisation of
$T^\ast G \times \mathbb{R}$.
\end{definition}

\begin{theorem}[Order of the contact LGVI]
\label{thm:order-contact-LGVI}
Let $L_h^{\mathrm e}$ be the exact discrete contact Lagrangian defined in
\eqref{eq:exact_contact_Ld}, and let $L_d$ be a discrete Lagrangian such that
\begin{equation}
  L_d(g_k,g_{k+1},z_k)
  = L_h^{\mathrm e}(g_k,g_{k+1},z_k) + \mathcal{O}(h^{r+1}),
  \qquad h \to 0,
  \label{eq:Ld-order-cond}
\end{equation}
uniformly in $(g_k,g_{k+1},z_k)$ on compact subsets.

Then the discrete flow generated by $L_d$ via the discrete Herglotz
principle (and its associated contact generating function $S_d$) yields an
order-$r$ approximation of the continuous Euler--Poincar\'e--Herglotz flow.  
In particular, the global error in $(g_k,p_k,z_k)$ over a fixed interval
$T = Nh$ is $\mathcal{O}(h^r)$.
\end{theorem}

\begin{proof}
We argue in two steps. First we show that the discrete flow induced by
$L_d$ has local truncation error $\mathcal{O}(h^{r+1})$ with respect to
the exact flow of the Euler--Poincar\'e--Herglotz equations. Then we invoke
standard stability results for one--step methods to obtain a global
$\mathcal{O}(h^r)$ error bound.

\medskip

\emph{Step 1: Local error via the exact discrete Lagrangian.}
Let $X_H$ be the contact Hamiltonian vector field on
$T^\ast G\times\R$ associated with the contact Lagrangian $L$ and let
$\Phi_H^t$ denote its flow.  By Proposition~\ref{prop:exact-contact-Lie},
the exact discrete contact Lagrangian $L_h^{\mathrm e}$ generates, via the
discrete Herglotz principle, a discrete map
\[
  F_h^{\mathrm e} : T^\ast G\times\R \to T^\ast G\times\R
\]
such that $F_h^{\mathrm e} = \Phi_H^h$, i.e., $F_h^{\mathrm e}$ coincides with
the exact time--$h$ flow of the continuous Euler--Poincar\'e--Herglotz dynamics.

Let $F_h^d$ denote the discrete map on $T^\ast G\times\R$ produced by
the same discrete Herglotz principle but using $L_d$ instead of
$L_h^{\mathrm e}$, with associated type--II contact generating function $S_d$.
By assumption \eqref{eq:Ld-order-cond},
\[
  L_d(g_k,g_{k+1},z_k)
    = L_h^{\mathrm e}(g_k,g_{k+1},z_k) + \mathcal{O}(h^{r+1}),
  \qquad h\to 0,
\]
uniformly on compact subsets. Since the definition of the type--II generating function simply adds the terms
$z_k$ and a bilinear pairing between $p_{k+1}$ and a discrete state increment
depending smoothly on $(g_k,g_{k+1})$, the same estimate carries over to the
discrete generating functions: if $S_h^{\mathrm e}$ denotes the (implicit)
type--II generating function associated with $L_h^{\mathrm e}$, then
\[
  S_d(g_k,p_{k+1},z_k)
    = S_h^{\mathrm e}(g_k,p_{k+1},z_k) + \mathcal{O}(h^{r+1}),
  \qquad h\to 0,
\]
again uniformly on compact sets.

The discrete update maps $F_h^{\mathrm e}$ and $F_h^d$ are defined
implicitly by the stationarity conditions of their respective generating
functions (discrete Herglotz principle / type--II formulation). In local
coordinates on $T^\ast G\times\R$, these stationarity equations can be
written in the form
\[
  \mathcal{F}^{\mathrm e}(y_k,y_{k+1},h) = 0,
  \qquad
  \mathcal{F}^d(y_k,y_{k+1},h) = 0,
\]
where $y_k = (g_k,p_k,z_k)$ and $y_{k+1}=(g_{k+1},p_{k+1},z_{k+1})$.
The functions $\mathcal{F}^{\mathrm e}$ and $\mathcal{F}^d$ are smooth in
their arguments and in $h$, and their difference satisfies
\[
  \mathcal{F}^d(y_k,y_{k+1},h)
    = \mathcal{F}^{\mathrm e}(y_k,y_{k+1},h)
      + \mathcal{O}(h^{r+1}),
\]
since the only difference comes from replacing $L_h^{\mathrm e}$ by
$L_d$ in the discrete variational principle (and hence in the corresponding
first variations).

For fixed $y_k$ lying on an exact trajectory, the implicit function
theorem applies to $\mathcal{F}^{\mathrm e}$ (hyperregularity of $L$
ensures that the Jacobian with respect to $y_{k+1}$ is invertible),
yielding a smooth dependence $y_{k+1} = F_h^{\mathrm e}(y_k)$.  Perturbing
$\mathcal{F}^{\mathrm e}$ by an $\mathcal{O}(h^{r+1})$ term
perturbs the solution $y_{k+1}$ by the same order. More precisely, one
obtains
\[
  F_h^{d}(y_k) - F_h^{\mathrm e}(y_k)
    = \mathcal{O}(h^{r+1}),
  \qquad h\to 0,
\]
uniformly for $y_k$ in compact sets. Since $F_h^{\mathrm e}=\Phi_H^h$,
this shows that the local truncation error of the method defined by
$L_d$ satisfies
\[
  F_h^{d}(y_k) - \Phi_H^h(y_k)
    = \mathcal{O}(h^{r+1}),
\]
for exact data $y_k = (g(t_k),p(t_k),z(t_k))$ at time $t_k=kh$.

\medskip

\emph{Step 2: Global error estimate.}
The discrete contact variational integrator defined by $L_d$ is a
one--step method on the smooth finite--dimensional manifold
$T^\ast G\times\R$.  The vector field $X_H$ is smooth, and for fixed
final time $T$ the exact solution $(g(t),p(t),z(t))$ remains in a compact
subset of $T^\ast G\times\R$ for $t\in[0,T]$. In any smooth local
trivialisation of $T^\ast G\times\R$ over such a compact set, the map
$F_h^d$ is uniformly Lipschitz in the state variable and smooth in $h$
for $h$ sufficiently small.

The standard convergence theorem for one--step methods for ODEs
(see, e.g., \cite{MarsdenWest2001,LeokShingel2012,marrero} in the variational
context, or any textbook result) states that a consistent one--step
method with local truncation error $\mathcal{O}(h^{r+1})$ and uniform
stability yields a global error $\mathcal{O}(h^r)$ over a fixed
interval $[0,T]$. Applying this to $F_h^d$ and the exact flow
$\Phi_H^h$ gives
\[
  \|(g_N,p_N,z_N) - (g(T),p(T),z(T))\|
    = \mathcal{O}(h^r),
  \qquad T=Nh,
  \qquad h\to 0,
\]
where the norm is taken in any smooth local trivialisation of
$T^\ast G\times\R$. Therefore, the
discrete flow generated by $L_d$ via the discrete Herglotz principle
has order $r$.
\end{proof}

\subsection{Retraction-based contact LGVI}

In practice, the exact discrete contact Lagrangian is replaced by a computable 
approximation constructed from a \emph{retraction map} $R : \fg \to G$ such that $R(0)=e$ and $\mathrm{d}R_0 = \mathrm{id}$, which locally approximates the exponential map.

Given consecutive states $(g_k,g_{k+1})\in G\times G$ and a step size $h$ 
small enough so that $g_k^{-1}g_{k+1}$ lies in the domain of $R^{-1}$, 
we define the discrete Lie--algebra increment
\begin{equation}
    \xi_k := \frac{1}{h} \, R^{-1}\!\bigl(g_k^{-1}g_{k+1}\bigr) \in \fg.
    \label{eq:LGVI_xi}
\end{equation}

A first--order retraction--based \emph{contact discrete Lagrangian} is
\begin{equation}
    L_d(g_k,g_{k+1},z_k)
      := h\,L\!\bigl(g_k,\xi_k,z_k\bigr),
    \label{eq:Ld_retraction}
\end{equation}
with higher--order midpoint or trapezoidal variants obtained by evaluating
$L$ at $(g_{k+1},\xi_k,z_k)$ or at a midpoint of the form
\[
  \bigl(g_k\,R(\tfrac12 h\xi_k),\,\xi_k,\,
        z_k+\tfrac12 h\,L(g_k,\xi_k,z_k)\bigr).
\]

We lift the construction to the extended space 
$M := T^\ast G\times\R$, endowed with the discrete left--trivialized 
contact one--form
\begin{equation}
    \Theta_k := \mathrm{d}z_k - \bigl\langle p_k, g_k^{-1}\mathrm{d}g_k\bigr\rangle,
    \label{eq:Theta-Lie-discrete}
\end{equation}
which is the discrete analogue of the canonical contact form 
$\theta = \mathrm{d}z - \langle p, g^{-1}\mathrm{d}g\rangle$
in the Lie--group setting.

Given $L_d$, we define the associated \emph{discrete contact generating 
function} of type II by
\begin{equation}
    S_d(g_k,p_{k+1},z_k)
      := z_k + L_d(g_k,g_{k+1},z_k)
         + \big\langle p_{k+1}, R^{-1}(g_k^{-1}g_{k+1}) \big\rangle,
    \label{eq:Sd-Lie}
\end{equation}
where $g_{k+1}$ is viewed as a function of $(g_k,p_{k+1},z_k)$ through the 
critical point conditions and $R^{-1}(g_k^{-1}g_{k+1})\in\mathfrak{g}$.

\begin{theorem}[Contact LGVI on a Lie group]
\label{thm:contact_LGVI}
Let $L_d$ be defined by \eqref{eq:Ld_retraction} and $S_d$ by 
\eqref{eq:Sd-Lie}.  
Consider the discrete map $\Psi_k : (g_k,p_{k+1},z_k) \longmapsto (g_{k+1},p_k,z_{k+1})$, where $(g_{k+1},p_k,z_{k+1})$ is determined implicitly by
\begin{equation}
    \partial_{g_k}S_d = 0,\qquad
    \partial_{p_{k+1}}S_d = 0,
    \qquad
    z_{k+1}=z_k + L_d(g_k,g_{k+1},z_k).
    \label{eq:crit-Sd-Lie}
\end{equation}
Then $\Psi_k$ is a strict discrete contactomorphism:
\begin{equation}
    \Psi_k^\ast \Theta_{k+1}
      = \Theta_k,
    \label{eq:discrete-contactomorphism-Lie}
\end{equation}
where $\Theta_k$ is given by \eqref{eq:Theta-Lie-discrete}.
In particular, the LGVI preserves the discrete contact structure on 
$T^\ast G\times\R$ exactly.
\end{theorem}

\begin{proof}
Differentiating \eqref{eq:Sd-Lie} and using the definition of $\Theta_k$ in 
\eqref{eq:Theta-Lie-discrete} yields, on the extended space with coordinates
$(g_k,p_k,z_k,g_{k+1},p_{k+1},z_{k+1})$, the discrete contact identity
\begin{equation}
    \mathrm{d}S_d
      = \Theta_k 
        + \big\langle p_{k+1}, g_{k+1}^{-1}\mathrm{d}g_{k+1}\big\rangle
        - \Theta_{k+1},
    \label{eq:dSd-identity-Lie}
\end{equation}
which is the Lie--group analogue of the identity 
\eqref{eq:dSd-identity_explicit} in Section~\ref{sec:discrete_contact_pmp}.

On the graph of the update map $\Psi_k$, the variables
$(g_{k+1},p_k,z_{k+1})$ are expressed in terms of $(g_k,p_{k+1},z_k)$ via
\eqref{eq:crit-Sd-Lie} and the recursion for $z_k$.  Restricting
\eqref{eq:dSd-identity-Lie} to this graph, and using the relation $z_{k+1} = z_k + L_d(g_k,g_{k+1},z_k)$ to eliminate $\mathrm{d}z_{k+1}$, we obtain an identity of the form
\[
    \mathrm{d}S_d
      = \Theta_k - \Psi_k^\ast\Theta_{k+1},
\]
where $S_d$ is now viewed as a function of the independent variables
$(g_k,p_{k+1},z_k)$.

Moreover, along a discrete trajectory generated by $S_d$ the stationarity
conditions \eqref{eq:crit-Sd-Lie} imply that the pullback of $\mathrm{d}S_d$ by $\Psi_k$ vanishes, $\Psi_k^\ast\mathrm{d}S_d = 0$.

Combining this with the previous identity yields
\[
  0 = \Psi_k^\ast\mathrm{d}S_d
    = \Theta_k - \Psi_k^\ast\Theta_{k+1},
\]
and hence $\Psi_k^\ast\Theta_{k+1}=\Theta_k$, which proves 
\eqref{eq:discrete-contactomorphism-Lie}.
\end{proof}

\begin{remark}
Theorem~\ref{thm:contact_LGVI} is the contact analogue of Lie--group 
variational integrators as in \cite{MarsdenWest2001,LeokShingel2012}, now 
formulated on the extended contact manifold $T^\ast G\times\R$.  
It provides the geometric backbone for structure--preserving discrete PMP 
schemes on Lie groups, improving on purely coordinate--based discrete PMP 
formulations.
\hfill$\diamond$
\end{remark}

\subsection{Discrete contact PMP on Lie groups}

We now combine the contact LGVI with the discrete contact PMP of 
Section~\ref{sec:discrete_contact_pmp} for a control system on $G$.
Let the controlled dynamics be given by
\begin{equation}
    g_{k+1} = g_k\,R\!\bigl(h\,f(g_k,u_k)\bigr),
    \label{eq:controlled-LGVI-dynamics}
\end{equation}
where $f:G\times U\to \fg$ is a left--trivialized control vector field and 
$R$ is a retraction as above.  

To avoid notational confusion with the discrete contact Lagrangian
$L_d(g_k,g_{k+1},z_k)$ of the LGVI, we denote the \emph{discrete running cost}
by
\[
  \ell_d : G\times U \to \R,
  \qquad
  (g_k,u_k) \mapsto \ell_d(g_k,u_k),
\]
and the terminal cost by $\Phi:G\to\R$.

We define the discrete contact Hamiltonian
\begin{equation}
    H_d(g_k,p_{k+1},z_k,u_k)
      := \Big\langle p_{k+1}, \frac{1}{h}\,R^{-1}\bigl(g_k^{-1}g_{k+1}\bigr)\Big\rangle
         + \ell_d(g_k,u_k)
      = \big\langle p_{k+1}, f(g_k,u_k)\big\rangle + \ell_d(g_k,u_k),
    \label{eq:Hd-Lie}
\end{equation}
where $g_{k+1}$ is given by \eqref{eq:controlled-LGVI-dynamics}, so that 
$R^{-1}(g_k^{-1}g_{k+1})=h\,f(g_k,u_k)$ and both expressions coincide.  

\begin{theorem}[Discrete contact PMP on a Lie group]
\label{thm:discrete-contact-PMP-Lie}
Let $\{u_k^\ast\}_{k=0}^{N-1}$ be optimal for the discrete problem on $G$, 
with associated trajectory $\{(g_k^\ast,z_k^\ast)\}$.  
Assume normality of the extremal.  
Then there exists a nontrivial costate sequence 
$\{p_k^\ast\}_{k=1}^N\subset\fg^\ast$ such that, for each $k=0,\dots,N-1$,
\begin{enumerate}
\item \textbf{State and cost updates:}
\begin{subequations}
\begin{align}
  g_{k+1}^\ast &= g_k^\ast\,R\!\bigl(h\,f(g_k^\ast,u_k^\ast)\bigr), \\
  z_{k+1}^\ast &= z_k^\ast + \ell_d(g_k^\ast,u_k^\ast).
\end{align}
\end{subequations}

\item \textbf{Adjoint recursion (left--trivialized):}
\begin{equation}
  p_k^\ast 
    = \Ad^\ast_{R(h f(g_k^\ast,u_k^\ast))}\,p_{k+1}^\ast
      + h\,\partial_g \ell_d(g_k^\ast,u_k^\ast),
  \label{eq:adjoint-Lie}
\end{equation}
where $\partial_g \ell_d(g_k^\ast,u_k^\ast)\in\fg^\ast$ denotes the left--trivial derivative of $\ell_d$ with respect to $g$.

\item \textbf{Terminal condition:}
\begin{equation}
  p_N^\ast = \mathrm{d}_g\Phi(g_N^\ast).
\end{equation}

\item \textbf{Stationarity:}
\begin{equation}
  u_k^\ast \in \arg\max_{u\in U}
    H_d(g_k^\ast,p_{k+1}^\ast,z_k^\ast,u),
  \label{eq:stationarity-Lie}
\end{equation}
with $H_d$ given by \eqref{eq:Hd-Lie}.
\end{enumerate}
Moreover, the optimal update map $\Psi_k:(g_k^\ast,p_{k+1}^\ast,z_k^\ast)\longmapsto(g_{k+1}^\ast,p_k^\ast,z_{k+1}^\ast)$ is a discrete contactomorphism of $(T^\ast G\times\R,\Theta_k)$ with respect to the discrete contact one--form \eqref{eq:Theta-Lie-discrete}, in the sense of
\eqref{eq:discrete-contactomorphism-Lie}.
\end{theorem}

\begin{proof}
The proof follows the same finite--dimensional Lagrange multiplier argument as 
in Theorem~\ref{thm:discrete_contact_pmp}, now in local trivializations of $G$.
The discrete dynamics \eqref{eq:controlled-LGVI-dynamics} play the role of the 
constraint map, and the multipliers $\lambda_{k+1}$ live in the dual Lie 
algebra $\fg^\ast$ after identification via charts and left trivialization.  
The adjoint recursion \eqref{eq:adjoint-Lie} is the intrinsic, left--trivialized
version of the multiplier update, and the stationarity condition 
\eqref{eq:stationarity-Lie} is equivalent to $\partial_u H_d=0$, hence to the 
maximum condition under the standing compactness and continuity assumptions 
(cf.\ Theorem~\ref{thm:discrete_contact_pmp}).  
Finally, the contactomorphism property is inherited from 
Theorem~\ref{thm:contact_LGVI}, since the optimal update map is precisely the 
contact LGVI flow associated with the discrete generating function $S_d$ 
augmented with the control--dependent term $\ell_d(g_k,u_k)$ in the cost
accumulation.
\end{proof}

\begin{remark}
The PMP on matrix Lie groups in \cite{PhogatChatterjeeBanavar2016} provides a 
geometric discrete PMP without explicit reference to contact geometry.  
Theorem~\ref{thm:discrete-contact-PMP-Lie} extends this viewpoint by showing 
that, when the discrete dynamics are generated by a contact LGVI, the discrete 
PMP update is automatically a contactomorphism on $T^\ast G\times\R$.  
In particular, the adjoint recursion \eqref{eq:adjoint-Lie} is consistent with 
the left--trivialized PMP in \cite{PhogatChatterjeeBanavar2016}, but enjoys the 
additional property of preserving the discrete contact structure induced by the 
cost accumulation.
\hfill$\diamond$
\end{remark}

\begin{corollary}[Abnormal discrete extremals on $G$]
If the Lagrange multiplier associated with the cost functional vanishes 
(abnormal case), the discrete contact PMP on $G$ reduces to
\[
  p_k^\ast 
    = \Ad^\ast_{R(h f(g_k^\ast,u_k^\ast))}\,p_{k+1}^\ast,
  \qquad
  \bigl(\partial_u f(g_k^\ast,u_k^\ast)\bigr)^\ast p_{k+1}^\ast = 0,
\]
with no contribution from $\ell_d$ and $\Phi$.  
Such extremals depend only on the geometry of the constraint 
$g_{k+1}=g_k R\!\bigl(h f(g_k,u_k)\bigr)$ and are typically ruled out in 
practice by rank conditions on $\partial_u f$ which force triviality of the 
multiplier sequence.  
\hfill$\diamond$
\end{corollary}

\section{Controlled Lindblad dynamics for a qubit}
\label{sec:lindblad_qubit}

We now specialize the discrete contact PMP of 
Section~\ref{sec:discrete_contact_pmp} to a controlled open quantum system:
a single qubit subject to Hamiltonian control and amplitude--damping 
Lindblad dissipation.  
The state space is the compact convex manifold
\[
   Q
   =
   \bigl\{\rho\in\C^{2\times2}:\ \rho^\dagger=\rho,\ \rho\ge0,\ \tr(\rho)=1\bigr\},
\]
equivalently the closed Bloch ball in $\R^3$.


Let $\rho(t)\in Q$ denote the qubit density operator, written in the
computational basis $\{\ket{0},\ket{1}\}$.  
The dynamics of a Markovian open system with coherent control $u(t)\in\R$
and amplitude--damping noise are given by the Gorini--Kossakowski--Sudarshan--
Lindblad (GKSL) master equation
\begin{equation}
  \dot\rho(t)
   = \mathcal{L}_{u(t)}\bigl(\rho(t)\bigr)
   := -i\bigl[H\bigl(u(t)\bigr),\rho(t)\bigr] 
      + \mathcal{D}\bigl(\rho(t)\bigr),
  \label{eq:lindblad-continuous}
\end{equation}
where $\mathcal{L}_u$ is the Lindblad generator associated with the fixed
control value $u$, acting linearly on density operators.

\medskip

The coherent part is generated by the Hamiltonian 
\[
  H(u) = \frac{u}{2}\,\sigma_x,
  \qquad
  \sigma_x
    =
    \begin{pmatrix}
    0 & 1 \\
    1 & 0
    \end{pmatrix}.
\]
For the purposes of this section we neglect any drift Hamiltonian; adding a
fixed drift $H_0$ does not affect the structure of the PMP.

\medskip

The dissipative part corresponds to the standard amplitude--damping channel
with Lindblad operator
\[
  L = \sqrt{\gamma}\,\sigma_-,
  \qquad
  \gamma>0,
\]
where
\[
  \sigma_- := \ket{0}\bra{1}
    =
    \begin{pmatrix}
    0 & 0 \\
    1 & 0
    \end{pmatrix}
\]
is the lowering (decay) operator from the excited state $\ket{1}$ to
the ground state $\ket{0}$.  The corresponding GKSL dissipator is
\[
   \mathcal{D}(\rho)
    = L\rho L^\dagger - \frac12\bigl\{L^\dagger L,\rho\bigr\}.
\]
Thus the dynamics \eqref{eq:lindblad-continuous} preserve positivity and
trace, as required for physical density operators.

\medskip

We consider the following Bolza--type optimal control problem: for a fixed
horizon $T>0$, minimize the cost functional
\begin{equation}
  J(u)
   = \int_0^T \alpha\,u^2(t)\,\mathrm{d}t + \Phi\bigl(\rho(T)\bigr),
  \qquad
  \alpha>0,
  \label{eq:lindblad-cost-continuous}
\end{equation}
subject to \eqref{eq:lindblad-continuous}.  

The terminal cost is chosen as 
\[
  \Phi(\rho)
    = 1-\tr\bigl(\rho\,\rho_{\mathrm{target}}\bigr),
  \qquad
  \rho_{\mathrm{target}} = \ket{0}\bra{0}.
\]
Since $\tr(\rho\,\rho_{\mathrm{target}})=\bra{0}\rho\ket{0}$ is the fidelity
with the ground state, one has $\Phi(\rho)=0$ iff $\rho=\rho_{\mathrm{target}}$
and $\Phi(\rho)>0$ otherwise.  
The running cost $\alpha u^2(t)$ penalizes large control amplitudes and
ensures the normality of the optimal extremal.

It is convenient to rewrite the dynamics in Bloch vector form.  
Let $\sigma_y$ and $\sigma_z$ denote the remaining Pauli matrices,
\[
  \sigma_y =
    \begin{pmatrix}
      0 & -i \\
      i & 0
    \end{pmatrix},
  \qquad
  \sigma_z =
    \begin{pmatrix}
      1 & 0 \\
      0 & -1
    \end{pmatrix},
\]
and write the density operator as
\[
  \rho
    = \frac12\bigl(I + r\cdot\sigma\bigr)
    = \frac12\Bigl(I + r_x\sigma_x + r_y\sigma_y + r_z\sigma_z\Bigr),
\]
where $I$ is the $2\times2$ identity, 
$r=(r_x,r_y,r_z)^\top\in\R^3$ is the Bloch vector, and 
$\sigma=(\sigma_x,\sigma_y,\sigma_z)$.
Positivity and trace imply $\|r\|\le1$, so the state space is the closed
Bloch ball $\mathbb{B}^3 = \{\,r\in\R^3:\ \|r\|\le1\,\}$.

Substituting this parametrization into \eqref{eq:lindblad-continuous} and
using standard Pauli algebra, one obtains the affine control system
\begin{equation}
  \dot r(t) = A\,r(t) + b + u(t)\,B\,r(t),
  \label{eq:bloch-affine-control}
\end{equation}
with matrices
\[
  A
  =
  \begin{pmatrix}
   -\dfrac{\gamma}{2} & 0 & 0 \\
   0 & -\dfrac{\gamma}{2} & 0 \\
   0 & 0 & -\gamma
  \end{pmatrix},
  \qquad
  b
  =
  \begin{pmatrix}
  0 \\[1mm] 0 \\[1mm] \gamma
  \end{pmatrix},
  \qquad
  B
  =
  \begin{pmatrix}
  0 & 0 & 0 \\
  0 & 0 & -1 \\
  0 & 1 & 0
  \end{pmatrix}.
\]

Equivalently,
\begin{equation}
\begin{aligned}
  \dot r_x &= -\frac{\gamma}{2}\,r_x,\\[1mm]
  \dot r_y &= -\frac{\gamma}{2}\,r_y - u\,r_z,\\[1mm]
  \dot r_z &= -\gamma\,(r_z-1) + u\,r_y
           = -\gamma\,r_z + \gamma + u\,r_y.
\end{aligned}
\label{eq:bloch-components}
\end{equation}
The term $A r + b$ describes amplitude--damping relaxation toward the ground
state $r_{\mathrm{target}}=(0,0,1)^\top$, while $u B r$ is the Hamiltonian
rotation generated by $H(u)=(u/2)\sigma_x$.

\medskip

In Bloch coordinates, the cost functional \eqref{eq:lindblad-cost-continuous}
becomes
\[
  J(u) = \int_0^T \alpha\,u^2(t)\,\mathrm{d}t
         + \Phi_{\mathrm{Bloch}}\bigl(r(T)\bigr),
\]
where a direct computation shows that
\[
  \Phi_{\mathrm{Bloch}}(r)
    = 1 - \tr\!\left(
        \frac12\bigl(I + r\cdot\sigma\bigr)\,\ket{0}\bra{0}
      \right)
    = 1 - \frac12\bigl(1 + r_z\bigr)
    = \frac12\bigl(1 - r_z\bigr).
\]
Thus, up to an affine rescaling, the terminal cost simply penalizes deviation 
of the $z$--component from its target value $r_z=1$.

\medskip

From the viewpoint of Section~\ref{sec:discrete_contact_pmp}, we may regard
either $\rho$ or its Bloch representation $r$ as the state variable $x$ in the
PMP, with dynamics $\dot x = f(x,u)$ given by \eqref{eq:lindblad-continuous} or 
\eqref{eq:bloch-affine-control}, respectively, and running cost
$L(x,u)=\alpha u^2$.  
In subsequent sections we discretize these dynamics using the 
contact Lie--group variational integrators of Section~\ref{sec:contact_LGVI},
and apply the discrete contact PMP to obtain structure--preserving optimal
control schemes for the controlled Lindblad qubit.

Throughout, we use a \emph{second--order accurate} contact LGVI and compare it
against a classical second--order Runge--Kutta--Munthe--Kaas (RKMK(2))
method.  
Thus any numerical differences arise purely from geometry---in particular,
preservation (or lack thereof) of discrete contact structure.


\medskip

Let $U(t)\in\SU(2)$ evolve by
\begin{equation}
  \dot U(t)= -iH(u(t))\,U(t),
  \qquad
  U(0)=I.
  \label{eq:unitary-evolution}
\end{equation}
For a constant control over a step $[t_k,t_{k+1}]$ of size $\Delta t$,
\begin{equation}
  U_{k+1}=\exp\!\bigl(-iH(u_k)\,\Delta t\bigr)\,U_k.
  \label{eq:LGVI-U}
\end{equation}
This is an exact Lie--group update of the type used in 
Runge--Kutta--Munthe--Kaas (RKMK) and commutator–free integrators
\cite{MuntheKaas1998,IserlesMuntheKaasNørsettZanna2000,CelledoniOwren2014}.
Given $U_{k+1}$, the unitary subflow on density operators is
\begin{equation}
   \rho_{k+1}^{(\mathrm u)}
   =U_{k+1}\,\rho_k\,U_{k+1}^\dagger .
   \label{eq:LGVI-rho-unitary}
\end{equation}


\medskip

Following the standard operator–sum (Kraus) representation of completely
positive trace–preserving (CPTP) quantum channels
\cite{Kraus1983}, any Markovian dissipative subflow admits an exact
expression of the form
\[
   \Phi(\rho)=\sum_j E_j \rho E_j^\dagger,
   \qquad
   \sum_j E_j^\dagger E_j = I.
\]
For amplitude–damping noise, the Kraus operators are
\begin{equation}
   E_0(\tau)=
   \begin{pmatrix}
      1 & 0 \\
      0 & \sqrt{1-p(\tau)}
   \end{pmatrix},
   \qquad
   E_1(\tau)=
   \begin{pmatrix}
      0 & \sqrt{p(\tau)} \\
      0 & 0
   \end{pmatrix},
   \qquad
   p(\tau)=1-e^{-\gamma\tau},
   \label{eq:Kraus-AD}
\end{equation}
which yield the exact dissipative subflow
\[
   \Phi_{\mathrm{AD}}^\tau(\rho)=
   E_0(\tau)\,\rho\,E_0(\tau)^\dagger \;+\;
   E_1(\tau)\,\rho\,E_1(\tau)^\dagger.
\]


To combine Hamiltonian and dissipative evolution, we use a symmetric
Strang splitting \cite{Strang1968}.  
If $\Phi_{\mathrm{AD}}^t$ and 
$\Phi_{\mathrm{u}}^t(\rho)=U(t)\rho U(t)^\dagger$ denote the exact
subflows, the Strang composition
\[
   \Phi_{\Delta t}^{(2)}
   =
   \Phi_{\mathrm{AD}}^{\Delta t/2}
   \circ
   \Phi_{\mathrm{u}}^{\Delta t}
   \circ
   \Phi_{\mathrm{AD}}^{\Delta t/2}
\]
yields a time–reversible second–order accurate approximation of the full
Lindblad evolution.

Using the exact Hamiltonian subflow 
$U_{k+1}=\exp(-iH(u_k)\Delta t)\,U_k$ and the analytic amplitude–damping
subflow, the discrete contact update reads:
\begin{subequations}
\label{eq:Strang-contact}
\begin{align}
  \rho_{k+\frac12}
    &= \Phi_{\mathrm{AD}}^{\Delta t/2}(\rho_k),\\[0.3em]
  \rho_{k+\frac12}^{(\mathrm{u})}
    &= U_{k+1}\,\rho_{k+\frac12}\,U_{k+1}^\dagger,\\[0.3em]
  \rho_{k+1}
    &= \Phi_{\mathrm{AD}}^{\Delta t/2}
        \!\left(\rho_{k+\frac12}^{(\mathrm{u})}\right).
\end{align}
\end{subequations}

This defines the discrete map
\[
   \rho_{k+1}=F^{(2)}(\rho_k,u_k),
\]
which is $C^\infty$ in $(\rho,u)$, symmetric, and second–order accurate:
\[
   \rho_{k+1}
   = \rho(t_k+\Delta t) + \mathcal{O}(\Delta t^2).
\]


We equip the space of Hermitian matrices with the real Hilbert--Schmidt
pairing $\langle A,B\rangle=\mathrm{Re}\,\tr(A^\dagger B)$, which is the
natural duality pairing compatible with the contact construction of
Sections~\ref{sec:discrete_contact_pmp}--\ref{sec:contact_LGVI}.
In Bloch coordinates this coincides with the Euclidean pairing on
$\mathbb{R}^3$.

The discrete extended dynamics are
\[
   \rho_{k+1}=F^{(2)}(\rho_k,u_k),
   \qquad
   z_{k+1}=z_k+L_d(\rho_k,u_k),
\]
where, consistently with the second--order accurate LGVI and with the
theory of exact and approximate discrete contact Lagrangians, the discrete
running cost is chosen as
\[
   L_d(\rho_k,u_k)
     = \alpha\,u_k^2\,\Delta t + \mathcal{O}(\Delta t^3).
\]
This corresponds to a second--order quadrature approximation of the
continuous Lagrangian $L(\rho,u)=\alpha u^2$ and ensures overall
second--order consistency of the discrete contact variational
integrator.

Following the type--II discrete contact Hamiltonian formalism of
Section~\ref{sec:discrete_contact_pmp}, the discrete contact Hamiltonian is
\begin{equation}
   H_d(\rho_k,P_{k+1},z_k,u_k)
      := \left\langle P_{k+1},\,F^{(2)}(\rho_k,u_k)\right\rangle
         \;+\; L_d(\rho_k,u_k),
   \label{eq:Hd-discrete-Lindblad}
\end{equation}
while the difference
\[
   F^{(2)}(\rho_k,u_k)-\rho_k
     = \dot\rho(t_k)\,\Delta t + \mathcal{O}(\Delta t^2)
\]
acts as a discrete approximation of the continuous state displacement,
fully analogous to the term $F(x_k,u_k)-x_k$ in the discrete PMP. With this sign convention, the discrete contact PMP is formulated as a
\emph{maximization} of $H_d$ with respect to the control, consistently with
Theorem~\ref{thm:discrete_contact_pmp}.

\begin{remark}[Hamiltonian sign convention for minimization problems]
\label{rem:Hd-sign-convention}
The discrete contact Hamiltonian \eqref{eq:Hd-discrete-Lindblad} has the generic
PMP structure
\[
  H_d(\rho_k,P_{k+1},z_k,u_k)
    = \bigl\langle P_{k+1},F^{(2)}(\rho_k,u_k)\bigr\rangle
      + L_d(\rho_k,u_k),
\]
so that minimization of the Bolza cost is achieved by
\emph{maximizing} $H_d$ with respect to $u_k$.

For numerical purposes it is often convenient to work with an equivalent
Hamiltonian of the form
\[
  \widetilde H_d(\rho_k,P_{k+1},u_k)
    := \bigl\langle P_{k+1},F^{(2)}(\rho_k,u_k)-\rho_k\bigr\rangle
       - L_d(\rho_k,u_k),
\]
which matches the usual ``$p\cdot(F-x)-L$'' convention for minimization problems.
Since the additional term $\langle P_{k+1},\rho_k\rangle$ does not depend on
$u_k$, both $H_d$ and $\widetilde H_d$ yield the same maximizers and therefore
the same PMP stationarity condition.  In the numerical implementation of
Section~\ref{sec:numerical} we use $\widetilde H_d$ in the Hamiltonian
maximization step, while the theoretical discrete PMP is written in terms of
$H_d$ in \eqref{eq:Hd-discrete-Lindblad}.
\hfill$\diamond$
\end{remark}

\medskip

A corresponding type--II discrete contact generating function is
\[
   S_d(\rho_k,P_{k+1},z_k,u_k)
     := z_k
        + L_d(\rho_k,u_k)
        + \big\langle P_{k+1},\,F^{(2)}(\rho_k,u_k)-\rho_k \big\rangle ,
\]
where $L_d$ is the discrete running cost introduced above.

Criticality of $S_d$ with respect to $(\rho_k,P_{k+1})$ reproduces the
discrete contact PMP conditions.  
Indeed, the equation $\partial_{P_{k+1}}S_d = F^{(2)}(\rho_k,u_k)-\rho_k$
recovers the discrete state update, while 
$\partial_{\rho_k}S_d$ yields the adjoint recursion.  
Similarly, the condition $\partial_{u_k}S_d = 0$ produces the discrete
stationarity requirement for the optimal control.  
Thus $S_d$ acts as a genuine type--II discrete contact generating function,
and generates the discrete contactomorphism defining the optimal update
\[
   (\rho_k,P_{k+1},z_k)\;\longmapsto\;
   (\rho_{k+1},P_k,z_{k+1}),
\]
in complete analogy with the Lie--group contact variational integrators
developed in Section~\ref{sec:contact_LGVI}.

\begin{proposition}[Discrete contact PMP for the Lindblad qubit]
\label{prop:dPMP-Lindblad}
Let $\{u_k^\ast\}_{k=0}^{N-1}$ be optimal for the discrete Bolza problem
associated with the data $(F^{(2)},L_d,\Phi)$, where
\[
  \rho_{k+1}=F^{(2)}(\rho_k,u_k),\qquad
  L_d(\rho_k,u_k)=\alpha\,u_k^2\,\Delta t+\mathcal{O}(\Delta t^3),
\]
and $\Phi(\rho)=1-\tr(\rho\,\rho_{\mathrm{target}})$.
Then there exists a nontrivial sequence of Hermitian costates
$\{P_k^\ast\}_{k=1}^N$ such that, for $k=0,\dots,N-1$,
\begin{align}
  &\text{\emph{state and cost updates:}}\nonumber\\[-0.4em]
  &\qquad
  \rho_{k+1}^\ast=F^{(2)}(\rho_k^\ast,u_k^\ast),\qquad
  z_{k+1}^\ast=z_k^\ast+L_d(\rho_k^\ast,u_k^\ast),\\[0.3em]
  &\text{\emph{adjoint recursion:}}\nonumber\\[-0.4em]
  &\qquad
  P_k^\ast
   = \partial_\rho H_d(\rho_k^\ast,P_{k+1}^\ast,z_k^\ast,u_k^\ast),
   \qquad k=0,\dots,N-1,\\[0.3em]
  &\text{\emph{terminal condition:}}\nonumber\\[-0.4em]
  &\qquad
   P_N^\ast=-\rho_{\mathrm{target}},\\[0.3em]
  &\text{\emph{stationarity:}}\nonumber\\[-0.4em]
  &\qquad
   u_k^\ast\in\arg\max_{u\in U}
      H_d(\rho_k^\ast,P_{k+1}^\ast,z_k^\ast,u),
   \qquad k=0,\dots,N-1,
\end{align}
where the discrete contact Hamiltonian $H_d$ is given by
\eqref{eq:Hd-discrete-Lindblad}.

In particular, with the convention
\eqref{eq:Hd-discrete-Lindblad} the discrete PMP for the Lindblad qubit
is naturally written as a \emph{maximization} of $H_d$ (or, equivalently,
of $\widetilde H_d$ in Remark~\ref{rem:Hd-sign-convention}) for a minimization
problem in Bolza form.

\end{proposition}

\begin{proof}
We apply the discrete contact PMP of
Section~\ref{sec:discrete_contact_pmp} with state manifold
$Q$ equal to the convex set of qubit density operators, endowed with
the real Hilbert--Schmidt pairing $\langle A,B\rangle=\mathrm{Re}\,\tr(A^\dagger B)$. By identifying covectors with Hermitian matrices via this pairing,
a discrete covector $p_k\in T_{\rho_k}^\ast Q$ corresponds to a
Hermitian matrix $P_k$ such that
\(
  p_k(\delta\rho)=\langle P_k,\delta\rho\rangle
\)
for all Hermitian perturbations $\delta\rho$. 

The discrete dynamics have the form
\[
  \rho_{k+1}=F^{(2)}(\rho_k,u_k),
  \qquad
  z_{k+1}=z_k+L_d(\rho_k,u_k),
\]
with $F^{(2)}$ the second--order Strang--split Lindblad integrator
defined in \eqref{eq:Strang-contact} and $L_d$ a second--order
approximation of the running cost
$\int_{t_k}^{t_{k+1}}\alpha u^2(t)\,\mathrm{d}t$.

The general discrete contact PMP (Theorem~\ref{thm:discrete_contact_pmp})
applied to the triple $(F^{(2)},L_d,\Phi)$ yields the following conditions.
First, the state and cost updates are those already stated above.  
Second, the adjoint variables satisfy the recursion
\[
  p_k
    = \partial_\rho L_d(\rho_k,u_k)
      + \bigl(\partial_\rho F^{(2)}(\rho_k,u_k)\bigr)^\ast p_{k+1},
\]
where $(\cdot)^\ast$ denotes the adjoint with respect to the
Hilbert--Schmidt pairing.  
Third, the costate at the final time is fixed by the terminal condition
\(
  p_N=\mathrm{d}_\rho\Phi(\rho_N).
\)
Finally, the optimal controls satisfy the discrete stationarity condition
\(
  u_k\in\arg\max_{u\in U} H_d(\rho_k,p_{k+1},z_k,u),
\)
where the discrete contact Hamiltonian takes the form
\[
  H_d(\rho_k,p_{k+1},z_k,u_k)
    = \langle p_{k+1},F^{(2)}(\rho_k,u_k)\rangle
      + L_d(\rho_k,u_k),
\]
which coincides with \eqref{eq:Hd-discrete-Lindblad} when written in the
matrix notation $p_{k+1}\leftrightarrow P_{k+1}$.

Since $L_d$ does not depend on $\rho_k$, we have
$\partial_\rho L_d(\rho_k,u_k)=0$, and the adjoint recursion simplifies to
\begin{equation}
   P_k=(\partial_\rho F^{(2)}(\rho_k,u_k))^\ast P_{k+1}.
   \label{eq:adjoint-F2-general}
\end{equation}
This is precisely the form stated in the proposition,
because
\[
  \partial_\rho H_d(\rho_k,P_{k+1},z_k,u_k)
    = (\partial_\rho F^{(2)}(\rho_k,u_k))^\ast P_{k+1},
\]
by differentiating \eqref{eq:Hd-discrete-Lindblad} with respect to $\rho_k$.

For the terminal condition, we use
\(
  \Phi(\rho)=1-\tr(\rho\,\rho_{\mathrm{target}})
\).
A variation $\delta\rho$ gives
\[
  \delta\Phi(\rho)
    = -\tr(\delta\rho\,\rho_{\mathrm{target}})
    = -\langle \rho_{\mathrm{target}},\delta\rho\rangle,
\]
so, with respect to the Hilbert--Schmidt pairing,
\(
  \mathrm{d}_\rho\Phi(\rho)=-\rho_{\mathrm{target}}
\),
and hence
\(
  P_N^\ast = -\rho_{\mathrm{target}}
\).

Finally, the stationarity condition is exactly the specialization of the
general maximization condition
$u_k^\ast\in\arg\max_{u\in U} H_d(\rho_k^\ast,P_{k+1}^\ast,z_k^\ast,u)$
to the Hamiltonian \eqref{eq:Hd-discrete-Lindblad}.
\end{proof}

Since $L_d$ is independent of $\rho$, the adjoint recursion can be written
directly as
\[
   P_k^\ast=(\partial_\rho F^{(2)}(\rho_k^\ast,u_k^\ast))^\ast P_{k+1}^\ast.
\]
For the Strang--split update \eqref{eq:Strang-contact}, the map
$F^{(2)}$ is the composition
\[
   F^{(2)}(\cdot,u_k)
     = \Phi_{\mathrm{AD}}^{\Delta t/2}
       \circ \mathcal{U}_{u_k}
       \circ \Phi_{\mathrm{AD}}^{\Delta t/2}(\cdot),
\]
where
\(
  \mathcal{U}_{u_k}(\rho)
    = U_k\,\rho\,U_k^\dagger
\)
with
\(
  U_k:=\exp(-iH(u_k)\Delta t)
\),
and $\Phi_{\mathrm{AD}}^\tau$ is the amplitude--damping channel from
\eqref{eq:Kraus-AD}.
The derivative of a linear CPTP map is the map itself, and its adjoint
with respect to the Hilbert--Schmidt pairing is the dual channel
$\Phi_{\mathrm{AD}}^{\tau\,\ast}$.  
For the unitary conjugation $\mathcal{U}_{u_k}$, the adjoint is
$X\mapsto U_k^\dagger X U_k$.  
Therefore, the adjoint recursion \eqref{eq:adjoint-F2-general} becomes
\begin{equation}
  P_k^\ast
    = \Phi_{\mathrm{AD}}^{\Delta t/2\,\ast}
        \!\left(
           U_k^\dagger\,
           \Phi_{\mathrm{AD}}^{\Delta t/2\,\ast}(P_{k+1}^\ast)\,
           U_k
        \right),
  \qquad
  U_k=\exp(-iH(u_k^\ast)\Delta t),
  \label{eq:adjoint-Strang-RKMK}
\end{equation}
which is the explicit backward propagation of the costate through the
Strang--split CPTP dynamics.

\begin{remark}
The map $(\rho_k,P_{k+1},z_k)\to(\rho_{k+1},P_k,z_{k+1})$ defined implicitly by the type--II generating function
\[
   S_d(\rho_k,P_{k+1},z_k,u_k)
     := z_k
        + L_d(\rho_k,u_k)
        + \big\langle P_{k+1},F^{(2)}(\rho_k,u_k)-\rho_k\big\rangle
\]
is a strict discrete contactomorphism of the contact manifold
$(T^\ast Q\times\R,\Theta)$, with
$\Theta=\mathrm{d}z-\langle P,\mathrm{d}\rho\rangle$,
by Theorem~\ref{thm:contactomorphism}.  
Thus the Lindblad LGVI splitting scheme preserves both the quantum state
space (complete positivity and trace preservation) and the discrete
contact geometry induced by the cost accumulation.
\hfill$\diamond$
\end{remark}




To obtain a fair comparison with the contact LGVI of the previous section,
we now introduce a standard \emph{second--order}  
Runge--Kutta--Munthe--Kaas method (RKMK(2)) applied to the Lindblad vector 
field $\dot\rho=\mathcal{L}_{u_k}(\rho)$ on the state manifold $Q$.
Both schemes are Lie–group–based and second order, but—as we emphasize
throughout this subsection—their geometric properties differ fundamentally.

\medskip

In full analogy with the LGVI formulation
$\rho_{k+1}=F^{(2)}(\rho_k,u_k)$,
we write the RKMK(2) update as
\[
   \rho_{k+1}^{\mathrm{RKMK}}
      = F_{\mathrm{RKMK}}(\rho_k,u_k),
\]
where the map $F_{\mathrm{RKMK}}$ is defined as follows.
For a constant control $u_k$, let
\[
   \xi(\rho) := \mathcal{L}_{u_k}(\rho)
\]
denote the Lindblad generator viewed as a linear operator on $Q$.
The two–stage RKMK(2) integrator is:

\begin{enumerate}
\item \emph{Stage values:}
\[
   K_1 = \xi(\rho_k), \qquad
   K_2 = \xi\!\left(\exp\!\left(\tfrac12\Delta t\,K_1\right)\rho_k
                     \exp\!\left(-\tfrac12\Delta t\,K_1\right)\right).
\]

\item \emph{Update:}
\begin{equation}
   \rho_{k+1}^{\mathrm{RKMK}}
     = \exp\!\bigl(\Delta t\,K_2\bigr)\,\rho_k\,
       \exp\!\bigl(-\Delta t\,K_2\bigr).
   \label{eq:RKMK2-update}
\end{equation}
\end{enumerate}

Like the LGVI, this method uses a Lie–group exponential for the coherent
subdynamics and is formally second order.  
Unlike the LGVI, however, the dissipative dynamics are incorporated only
through the embedded algebra element $K_2$, and therefore complete positivity
and trace preservation are \emph{not} guaranteed.

\medskip

To place RKMK(2) in direct parallel with the contact LGVI, we endow it with
the \emph{ordinary} discrete PMP structure.  
The discrete running cost is chosen with the same second–order accuracy as
in the LGVI:
\[
   L_d^{(2)}(\rho_k,u_k)
      = \alpha\,u_k^2\,\Delta t + \mathcal{O}(\Delta t^3),
\]
which matches the approximation order of the continuous Lagrangian
$L(\rho,u)=\alpha u^2$.

\medskip

Mirroring the definition of the contact Hamiltonian $H_d$ used in the LGVI,
but \emph{without} the contact term $z$, we set
\begin{equation}
   H_d^{\mathrm{RKMK}}(\rho_k,P_{k+1},u_k)
      := \big\langle P_{k+1},\,F_{\mathrm{RKMK}}(\rho_k,u_k)\big\rangle
         + L_d^{(2)}(\rho_k,u_k).
   \label{eq:Hd-RKMK}
\end{equation}
This plays the same algebraic role as \eqref{eq:Hd-discrete-Lindblad},
but it is \emph{not} a contact Hamiltonian.

\medskip

In exact analogy with the LGVI update
$(\rho_k,z_k)\mapsto(\rho_{k+1},z_{k+1})$,
the RKMK extended dynamics read:
\[
   \rho_{k+1}=F_{\mathrm{RKMK}}(\rho_k,u_k),
   \qquad
   z_{k+1}=z_k+L_d^{(2)}(\rho_k,u_k).
\]

\medskip

Since $L_d^{(2)}$ is independent of $\rho_k$, the ordinary PMP yields
\[
   P_k
     = \partial_\rho H_d^{\mathrm{RKMK}}(\rho_k,P_{k+1},u_k)
     = \bigl(\partial_\rho F_{\mathrm{RKMK}}(\rho_k,u_k)\bigr)^\ast
        P_{k+1}.
\]
This expression is structurally identical to the LGVI adjoint recursion,
except that here the derivative acts on $F_{\mathrm{RKMK}}$ rather than on the
Strang–split CPTP map $F^{(2)}$.

\medskip

Here the parallelism with the LGVI stops:  
\[
   (\rho_k,P_{k+1},z_k)\longmapsto
   (\rho_{k+1},P_k,z_{k+1})
\]
is \emph{not} generated by a type–II contact generating function.
Consequently, the update is \emph{not} a discrete contactomorphism; the discrete contact form $\Theta=\mathrm{d}z-\langle P,\mathrm{d}\rho\rangle$
      is \emph{not} preserved; and complete positivity and trace preservation are not enforced by the
      integrator.

In contrast, the contact LGVI enforces all three properties exactly.

\begin{remark}
\label{rem:contrast-contact-RKMK}
Both integrators are second–order and both exploit Lie–group structure in the
unitary dynamics.  
However, the LGVI is built from a discrete contact Lagrangian and a
type–II discrete contact generating function, uses exact Kraus maps for
dissipation, and yields a strict discrete contactomorphism preserving
complete positivity, trace, and contact geometry.

The RKMK(2) method, while formally similar in order and Lie–group handling,
is not derived from any discrete Lagrangian or generating function.  
It therefore lacks exact dissipative structure, exact positivity,
and discrete contact preservation.

Thus the two methods share order of accuracy but differ fundamentally in
geometric fidelity—a distinction that becomes evident in the numerical
experiments of Section~\ref{sec:numerical}.
\hfill$\diamond$
\end{remark}

\section{Numerical results}
\label{sec:numerical}

In this section we illustrate the behaviour of the contact LGVI discretization
developed in Section~\ref{sec:lindblad_qubit} and compare it with a
second--order Runge--Kutta--type scheme (labelled \textsc{RK2}) applied directly
to the Lindblad master equation.  
Both methods have order~$2$; the comparison is therefore focused on
their \emph{geometric} properties (CPTP structure, contact geometry), rather
than on order of accuracy.

Throughout we consider the controlled amplitude--damping qubit of
Section~\ref{sec:lindblad_qubit}, with Hilbert space $\cH=\C^2$, Hamiltonian
$H(u)=\tfrac{u}{2}\sigma_x$ and dissipator
\(
   \cD(\rho) = L\rho L^\dagger-\tfrac12\{L^\dagger L,\rho\},
   \ L=\sqrt{\gamma}\,\sigma_-.
\)
Unless otherwise stated, the initial state is $\rho_0=\ket1\bra1$, and the
target state is $\rho_{\mathrm{target}}=\ket0\bra0$.

\subsection{Experimental set--up and performance metrics}
\label{subsec:setup_metrics}

The continuous dynamics are governed by the Lindblad equation
\begin{equation}
  \dot\rho(t)
   = \cL_{u(t)}(\rho(t))
   := -i[H(u(t)),\rho(t)] + \cD(\rho(t)).
\end{equation}
For the numerical experiments we fix a terminal time~$T$ and a uniform grid
$t_k=k\Delta t$, $k=0,\dots,N$, with $N=T/\Delta t$.  
The control is chosen as a smooth pulse
\begin{equation}
  u(t) = 4\sin\!\left(\frac{\pi t}{T}\right),
\end{equation}
and discretized as $u_k = u(t_k)$.


The ``contact'' scheme uses the second--order Strang splitting constructed in
Section~\ref{sec:lindblad_qubit}.  
On each time step the state update is
\begin{subequations}
\label{eq:contact_step_num}
\begin{align}
  \rho_{k+\frac12} &= \Phi_{\mathrm{AD}}^{\Delta t/2}(\rho_k),\\
  \rho_{k+\frac12}^{(\mathrm{u})}
    &= U_{k+1}\,\rho_{k+\frac12}\,U_{k+1}^\dagger,
    \qquad
    U_{k+1} = \exp\!\bigl(-iH(u_k)\,\Delta t\bigr),\\
  \rho_{k+1}
    &= \Phi_{\mathrm{AD}}^{\Delta t/2}\!\bigl(\rho_{k+\frac12}^{(\mathrm{u})}\bigr),
\end{align}
\end{subequations}
where $\Phi_{\mathrm{AD}}^\tau$ is the exact amplitude--damping channel with
Kraus operators~\eqref{eq:Kraus-AD}.  
Each step of this map is completely positive and trace preserving (CPTP) and
defines the discrete dynamics
$\rho_{k+1}=F_{\mathrm{C}}^{(2)}(\rho_k,u_k)$
used inside the contact LGVI framework.  
The extended variables $(P_k,z_k)$ are updated by the discrete contact PMP of
Section~\ref{sec:lindblad_qubit}, so that the extended map
$(\rho_k,P_{k+1},z_k)\mapsto(\rho_{k+1},P_k,z_{k+1})$ is a strict discrete
contactomorphism.

As non--contact reference we use an explicit second--order Runge--Kutta
method of Heun type applied directly to the Lindblad generator,
\begin{equation}
  \dot\rho = \cL_{u(t)}(\rho)
           = -i\bigl[H(u(t)),\rho\bigr] + \cD(\rho).
\end{equation}
On each step we compute
\begin{align*}
  k_1 &= \cL_{u_k}(\rho_k),\\
  k_2 &= \cL_{u_{k+\frac12}}
          \bigl(\rho_k + \tfrac{\Delta t}{2}k_1\bigr),
\end{align*}
and update
\begin{equation}
  \rho_{k+1}
    = \rho_k + \frac{\Delta t}{2}\bigl(k_1 + k_2\bigr).
  \label{eq:rk2_state_update}
\end{equation}
In the numerical experiments we set $u_{k+\frac12} = u_k$, a standard simplification that preserves the global second--order accuracy of
the Heun method.

This RK2 integrator is used purely as a non--geometric benchmark: it does
\emph{not} exploit the Kraus/CPTP structure of the Lindblad flow and is not
guaranteed to preserve positivity, unit trace, or the discrete contact form.
Accordingly, the extended variables $(P_k,z_k)$ are updated using the ordinary
discrete PMP (i.e.\ without a contact generating function), so the RK2
extended map is not a discrete contactomorphism.


To obtain an accuracy benchmark we compute a ``reference'' trajectory 
$\rho_{\mathrm{ref}}(t_k)$ using the contact LGVI integrator with a much
smaller time step $\Delta t_{\mathrm{ref}}=\Delta t/20$.  
For the time grids considered here this fine solution is visually
indistinguishable from the exact Lindblad evolution.


For a given trajectory $\{\rho_k\}$ we monitor:

\begin{itemize}
  \item \emph{Trace drift:}
  \begin{equation}
    \delta_{\mathrm{tr}}(k)
      := \bigl|\tr(\rho_k)-1\bigr|.
  \end{equation}
  For a CPTP integrator this quantity should remain at roundoff level.

  \item \emph{Positivity drift:}  
  Let $\lambda_{\min}(\rho_k)$ denote the smallest eigenvalue of $\rho_k$.  
  We define
  \begin{equation}
    \delta_{\mathrm{pos}}(k)
      := -\min\{0,\lambda_{\min}(\rho_k)\},
  \end{equation}
  which vanishes exactly when $\rho_k$ is positive semidefinite.

  \item \emph{Global error vs.\ reference:}
  \begin{equation}
    \varepsilon_{\mathrm{glob}}(k)
      := \bigl\|\rho_k - \rho_{\mathrm{ref}}(t_k)\bigr\|_F.
  \end{equation}

  \item \emph{Contact--form defect:}  
  \begin{equation}
    \theta_k
      := (z_{k+1}-z_k)
         - \bigl\langle P_{k+1},\rho_{k+1}-\rho_k\bigr\rangle,
  \end{equation}
  where $\langle\cdot,\cdot\rangle$ denotes the Hilbert--Schmidt pairing.  
  For an exact discrete contactomorphism we expect $\theta_k\equiv 0$
  (up to floating--point roundoff).  
  Thus, for the contact LGVI we expect $\theta_k$ to remain at machine
  precision, while for the non--contact scheme $\theta_k$ is unconstrained
  and typically grows with~$\Delta t$ and~$k$.
\end{itemize}

\subsection{Short–horizon accuracy and structure preservation}
\label{subsec:short_horizon}

We first consider a moderately short horizon with weak damping $T = 10, \,
   \Delta t = 0.01, \,
   \gamma = 1$. Both integrators are initialized at $\rho_0=\ket1\bra1$ and driven by the same
control pulse $u(t)=4\sin(\pi t/T)$.  

The ``contact'' scheme is the second–order Strang–split LGVI constructed in
Section~\ref{sec:lindblad_qubit}, which combines the exact unitary subflow
generated by $H(u)=\tfrac{u}{2}\sigma_x$ with the exact amplitude–damping
channel $\Phi_{\mathrm{AD}}^\tau$ in Kraus form.  
On each time step the state update is
\begin{subequations}
\label{eq:contact_step_num}
\begin{align}
  \rho_{k+\frac12} &= \Phi_{\mathrm{AD}}^{\Delta t/2}(\rho_k),\\
  \rho_{k+\frac12}^{(\mathrm{u})}
    &= U_{k+1}\,\rho_{k+\frac12}\,U_{k+1}^\dagger,
    \qquad
    U_{k+1} = \exp\!\bigl(-iH(u_k)\,\Delta t\bigr),\\
  \rho_{k+1}
    &= \Phi_{\mathrm{AD}}^{\Delta t/2}\!\bigl(\rho_{k+\frac12}^{(\mathrm{u})}\bigr),
\end{align}
\end{subequations}
where $\Phi_{\mathrm{AD}}^\tau$ is the exact amplitude–damping channel with
Kraus operators~\eqref{eq:Kraus-AD}.  
Each step of this map is completely positive and trace preserving (CPTP) and
defines the discrete dynamics
$\rho_{k+1}=F_{\mathrm{C}}^{(2)}(\rho_k,u_k)$
used inside the contact LGVI framework.  
The extended variables $(P_k,z_k)$ are updated by the discrete contact PMP of
Section~\ref{sec:lindblad_qubit}, so that the extended map
$(\rho_k,P_{k+1},z_k)\mapsto(\rho_{k+1},P_k,z_{k+1})$ is a strict discrete
contactomorphism.

As non–contact reference we use the explicit RK2 (Heun) method described in
Section~\ref{subsec:setup_metrics}.  
This method evolves the same Lindblad generator $\cL_u$, but it does not
preserve positivity or unit trace exactly, nor does it preserve the discrete
contact structure.

Figure~\ref{fig:bloch_traj_T10} shows the Bloch trajectories of both schemes,
obtained from
$\rho_k=\tfrac12(I+x_k\sigma_x+y_k\sigma_y+z_k\sigma_z)$.
The contact LGVI and RK2 trajectories are qualitatively similar and remain
close to each other over the horizon; in both cases the dynamics display the
expected contraction towards the ground state along the $z$–axis due to
amplitude damping.

\begin{figure}[h!]
  \centering
  \includegraphics[width=0.7\linewidth]{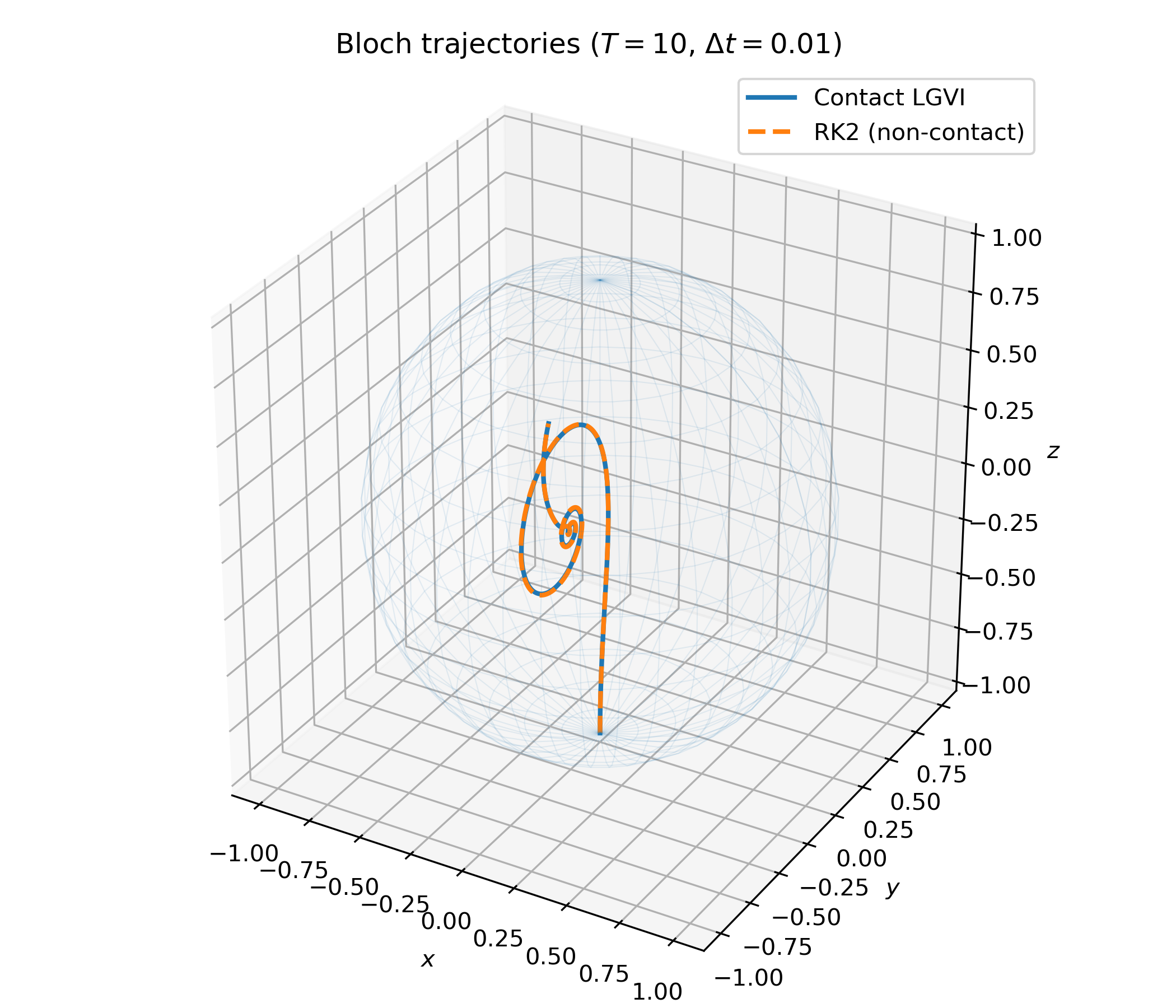}
  \caption{Bloch trajectories for $T=10$, $\Delta t=0.01$.
  Both schemes start from $\rho_0=\ket1\bra1$ and are driven by the same
  sinusoidal control.  
  The trajectories remain close and confined to the Bloch ball; the dynamics
  show the expected contraction towards the ground state.}
  \label{fig:bloch_traj_T10}
\end{figure}

Figures~\ref{fig:trace_drift_T10} display the trace drift and positivity drift
on a logarithmic scale.  
As expected, the contact LGVI preserves the trace at roundoff level and remains
exactly positive semidefinite, with deviations compatible with machine
precision.  
The RK2 method exhibits small but visible trace and positivity violations,
which remain modest on this time scale but are already systematic.

\begin{figure}[h!]
  \centering
  \includegraphics[width=0.5\linewidth]{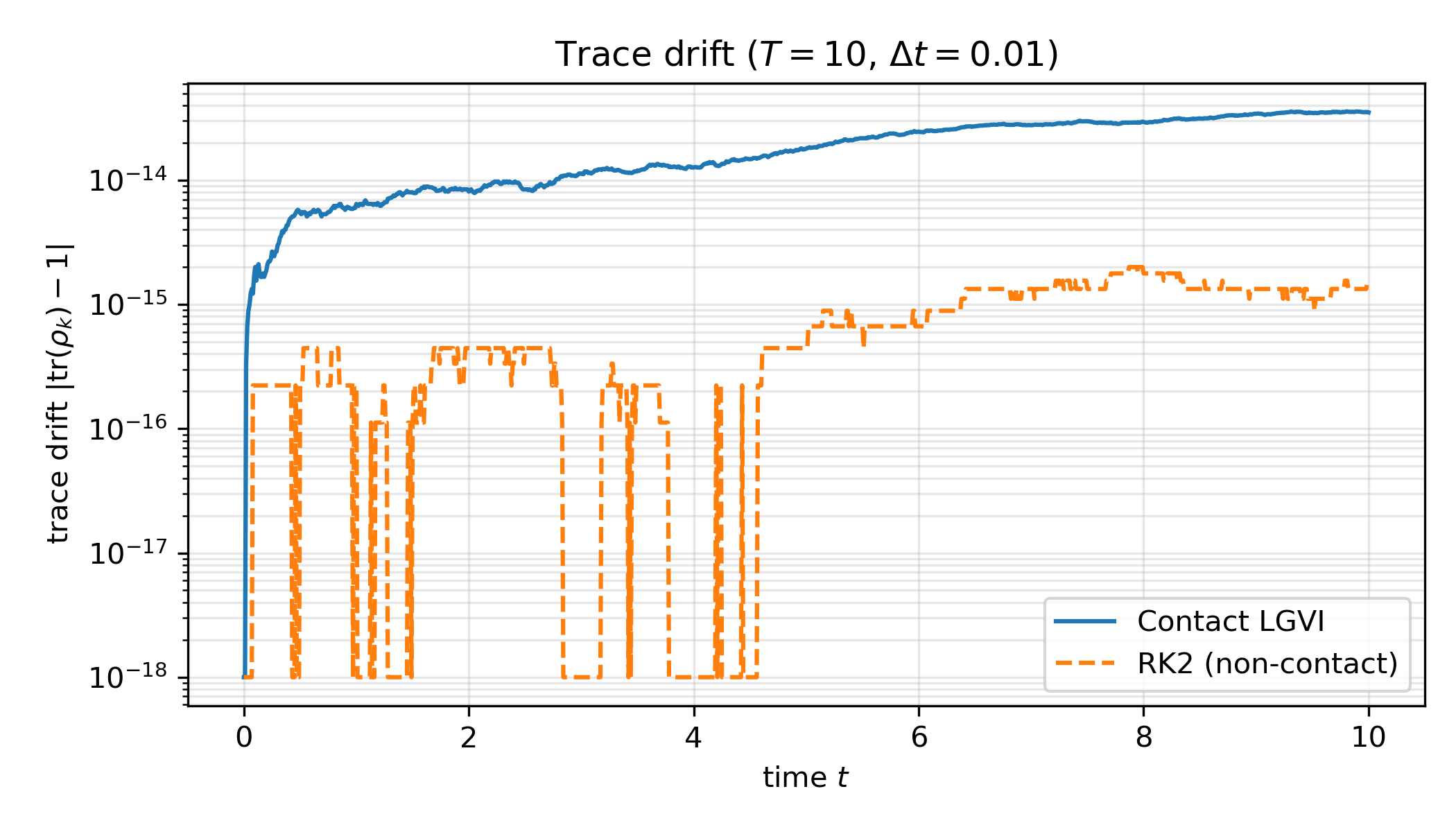}\includegraphics[width=0.5\linewidth]{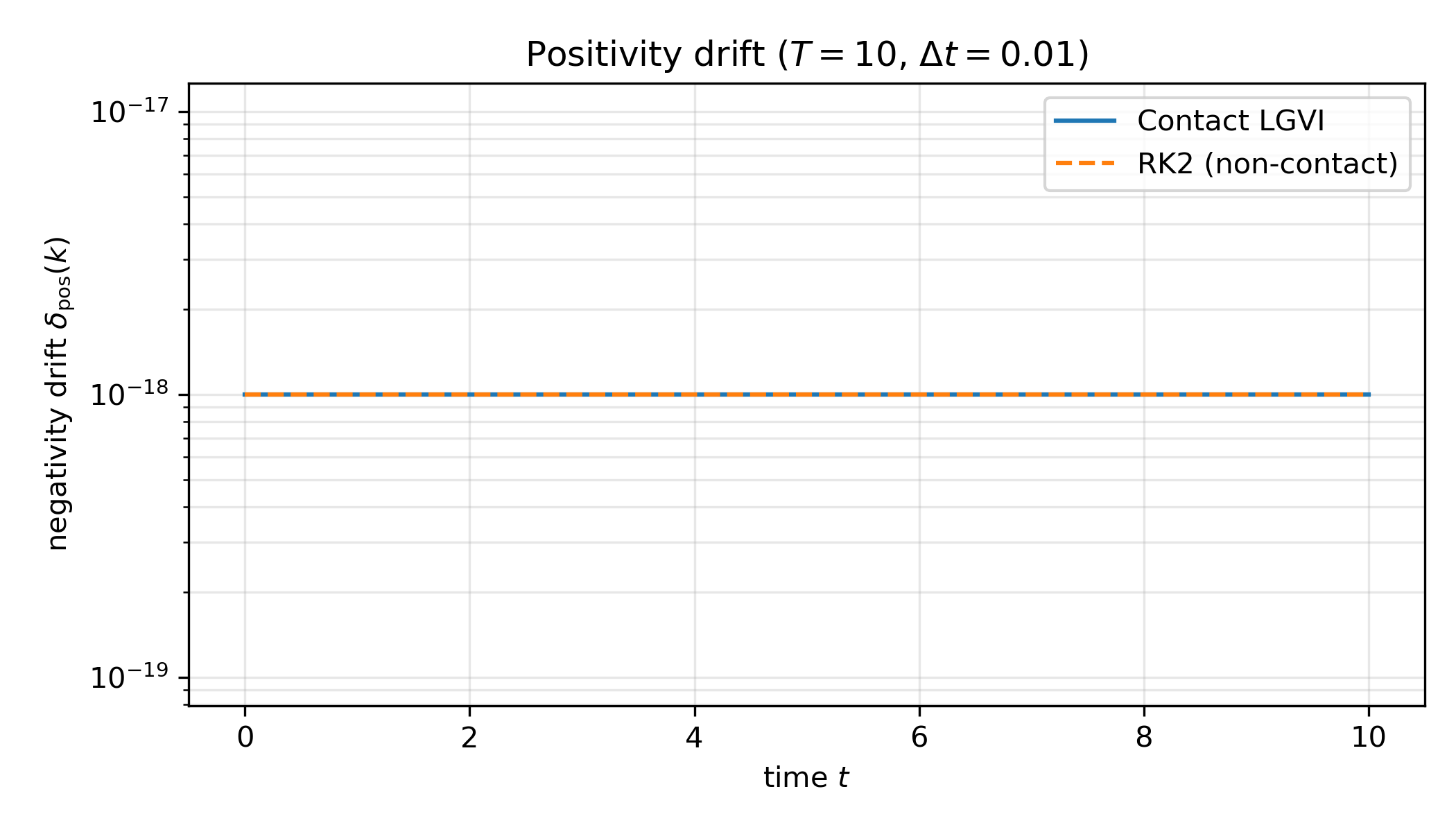}
  \caption{Left: Trace drift $|\tr(\rho_k)-1|$ for $T=10$, $\Delta t=0.01$.
  The contact LGVI keeps the trace at roundoff level, while RK2 exhibits a
  small but systematic drift, reflecting its lack of exact CPTP structure.
  Right: Positivity violation $\delta_{\mathrm{pos}}(k)$ for $T=10$,
  $\Delta t=0.01$.
  The contact LGVI remains positive semidefinite to machine precision, whereas
  RK2 displays slight negativity in the eigenvalues of $\rho_k$.}
  \label{fig:trace_drift_T10}
\end{figure}

Finally, Figure~\ref{fig:local_error_T10} shows the global error
$\varepsilon_{\mathrm{glob}}(k)$ with respect to the fine–step contact LGVI
reference.  
Both methods display the expected second–order behaviour, with error growing
smoothly over time.  
The contact LGVI achieves slightly smaller error constants, while at the same
time preserving CPTP structure exactly.

\begin{figure}[h!]
  \centering
  \includegraphics[width=0.7\linewidth]{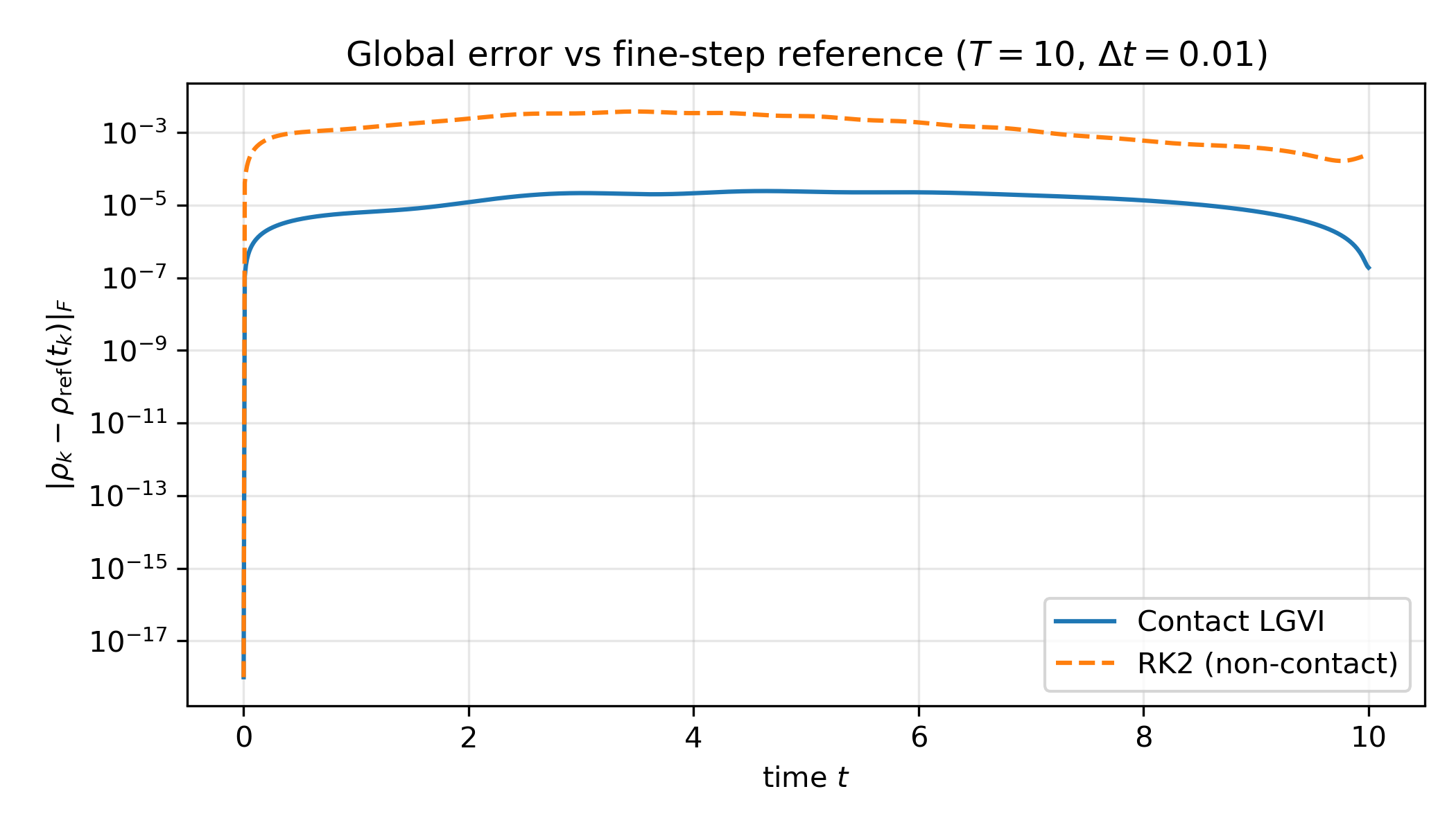}
  \caption{Global error
  $\varepsilon_{\mathrm{glob}}(k)
   =\|\rho_k-\rho_{\mathrm{ref}}(t_k)\|_F$
  for $T=10$, $\Delta t=0.01$, using a fine–step contact LGVI trajectory as
  reference.  
  Both schemes exhibit second–order accuracy; the contact LGVI has slightly
  smaller error while maintaining exact CPTP structure.}
  \label{fig:local_error_T10}
\end{figure}

Overall, on short horizons both schemes are accurate and qualitatively
consistent, but only the contact LGVI is exactly CPTP and contact–geometric at
the discrete level.

\subsection{Long–horizon stability}
\label{subsec:long_horizon}

We now investigate numerical robustness on a longer horizon and under strong
damping, a regime where geometric drift is most pronounced.  
We take $T = 100, \,
  \Delta t = 0.01, \,
  \gamma = 10$, leading to $N = 10^4$ steps.  
The initial state and control pulse are the same as in the previous experiment. For both schemes we compute the trace drift, positivity drift and
contact–form defect.  

The following table summarizes the observed maxima:
\begin{center}
\begin{tabular}{lccc}
\hline
 & $\displaystyle \max_k\delta_{\mathrm{tr}}(k)$
 & $\displaystyle \max_k\delta_{\mathrm{pos}}(k)$
 & $\displaystyle \max_k|\theta_k|$ \\
\hline
Contact LGVI 
 & $8.9\times 10^{-14}$ 
 & $0$
 & $1.8\times 10^{-2}$ \\
RK2 (non–contact) 
 & $3.6\times 10^{44}$ 
 & $8.9\times 10^{60}$ 
 & $1.7\times 10^{64}$ \\
\hline
\end{tabular}
\end{center}


Figure~\ref{fig:trace_drift_T100} (left) shows the trace drift
$\delta_{\mathrm{tr}}(k)=|\mathrm{tr}(\rho_k)-1|$.
The contact LGVI maintains the trace at roundoff level throughout the
$10^4$--step evolution.
In contrast, the non–contact RK2 scheme exhibits exponential blow–up,
ultimately losing $40$ orders of magnitude in trace preservation. Figure~\ref{fig:trace_drift_T100} (right) also displays  the positivity drift
$\delta_{\mathrm{pos}}(k)$.
The contact LGVI maintains positivity exactly up to floating--point noise,
as expected from its CPTP structure.  
The RK2 method, lacking Kraus or CPTP constraints, rapidly produces negative
eigenvalues whose magnitude eventually exceeds $10^{60}$.

\begin{figure}[h!]
  \centering
  \includegraphics[width=0.5\linewidth]{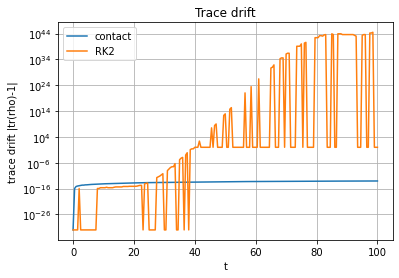}\includegraphics[width=0.5\linewidth]{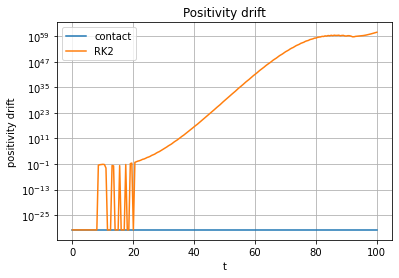}
  \caption{Left: Trace drift for $T=100$, $\Delta t=0.01$, $\gamma=10$.
  The contact LGVI remains at machine precision, whereas the RK2 method
  loses trace preservation catastrophically. Right: Positivity drift for $T=100$, $\Delta t=0.01$, $\gamma=10$.
  The contact LGVI stays positive semidefinite; RK2 produces massive
  violations of positivity}
  \label{fig:trace_drift_T100}
\end{figure}

Finally, Ffgure~\ref{fig:theta_T100} illustrates the magnitude of the contact–form
defect
\[
   \theta_k
      := (z_{k+1}-z_k)
         - \langle P_{k+1},\rho_{k+1}-\rho_k \rangle .
\]
The contact LGVI scheme exhibits the expected $\mathcal{O}(\Delta t^2)$ level
of defect (bounded by $\approx 10^{-2}$).  
The RK2 scheme, lacking a generating function or discrete contact structure, 
accumulates a defect that becomes astronomically large ($\sim 10^{64}$).

\begin{figure}[h!]
  \centering
  \includegraphics[width=0.72\linewidth]{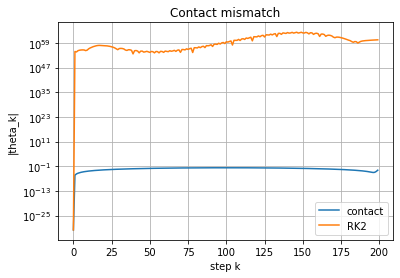}
  \caption{Contact–form defect $\theta_k$ for $T=100$, $\Delta t=0.01$,
  $\gamma=10$.  
  LGVI remains geometrically consistent; RK2 diverges dramatically.}
  \label{fig:theta_T100}
\end{figure}


These results show that the contact LGVI discretization preserves CPTP
structure and contact geometry even under strong dissipation and long horizons.
The non–contact RK2 method, although second–order accurate in principle,
becomes numerically nonphysical: it loses trace, breaks positivity, and drifts
away from the contact–form identity by dozens of orders of magnitude.

This experiment confirms that geometric preservation is essential for robust
long–horizon simulations of open quantum systems.

We evaluate robustness for a long horizon with strong dissipation 
($T=100$, $\Delta t=0.01$, $\gamma=10$).  
Across $10^4$ time steps, the contact LGVI remains numerically CPTP:
the trace stays at roundoff ($\sim10^{-13}$), positivity is preserved exactly, 
and the contact–form defect remains bounded at the expected 
$\mathcal{O}(\Delta t^{2})$ level ($\sim10^{-2}$).  
This behaviour is consistent with the theoretical geometric guarantees of the
scheme. In contrast, the non–contact RK2 discretization becomes rapidly nonphysical
over the same horizon.  
Trace errors grow to $\mathcal{O}(10^{44})$, positivity is violated with
negative eigenvalues of size $\mathcal{O}(10^{60})$, and the contact–form 
defect reaches $\mathcal{O}(10^{64})$.  
Thus, even though both integrators have formal order $2$, only the contact 
LGVI maintains geometric consistency over long horizons when simulating 
open quantum systems with dissipation.

\subsection{Optimal control with the discrete contact PMP}
\label{subsec:optimal_control}

We now assess the impact of the integrator choice on the solution of an
optimal control problem solved via a discrete Pontryagin Maximum Principle.
We consider the same controlled amplitude--damping qubit as in
Section~\ref{sec:lindblad_qubit}, with parameters $T = 3, \, \Delta t = 0.01, \, \gamma = 1,\,
   \alpha = 0.05$, and box-constrained controls $u_k\in[-u_{\max},u_{\max}]$ with
$u_{\max}=6$.  The initial and target states are
$\rho_0=\ket1\bra1$ and $\rho_{\mathrm{target}}=\ket0\bra0$.
The Bolza cost is
\[
  J(u) = \sum_{k=0}^{N-1} \alpha u_k^2\,\Delta t
         + \Phi(\rho_N), 
  \qquad
  \Phi(\rho)=1-\tr(\rho\,\rho_{\mathrm{target}}).
\]

At each iteration of the PMP shooting algorithm we propagate the state,
cost, and costate forward and backward in time, and update the control
via maximization of the discrete Hamiltonian $H_d(\rho_k,P_{k+1},u)$.
We compare two discretizations:

\begin{itemize}
  \item \textbf{Contact LGVI + discrete contact PMP}:  
        the state update is the CPTP Strang map
        $F^{(2)}_{\mathrm{C}}$,  
        the cost and costate follow the discrete contact relations derived in
        Section~\ref{sec:lindblad_qubit},  
        and the extended update
        $(\rho_k,P_{k+1},z_k)\mapsto(\rho_{k+1},P_k,z_{k+1})$
        is a strict \emph{discrete contactomorphism}.

  \item \textbf{RK2 (non-contact) + ordinary discrete PMP}:  
        the state is propagated by an explicit second--order Runge--Kutta
        scheme of Heun type applied to the Lindblad generator,  
        and the costate follows the corresponding RK2 adjoint recursion.  
        This scheme is \emph{not} CPTP and \emph{not} a contactomorphism.
\end{itemize}

Both PMP solvers start from the \emph{same} initial guess for the control,
\[
   u_k^{(0)}
      = 4\sin\!\left(\frac{\pi t_k}{T}\right), 
   \qquad t_k = k\Delta t,
\]
and use identical relaxation $\beta=0.5$.


\begin{figure}[h!]
  \centering
  \includegraphics[width=0.6\linewidth]{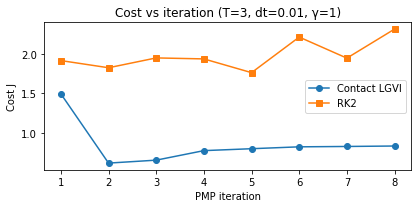}
  \caption{Evolution of the cost $J$ over the PMP iterations for $T=3$,
  $\Delta t=0.01$, $\gamma=1$, $\alpha=0.05$.
  The contact LGVI (solid line) displays a regular decrease and converges
  to a slightly lower cost, whereas the RK2 discretisation (dashed line)
  shows a more oscillatory behaviour and settles at a higher value of $J$.}
  \label{fig:ocp_cost_T3}
\end{figure}

Figure~\ref{fig:ocp_cost_T3} shows the evolution of the cost $J$ over the
first few shooting iterations for both discretisations.  
The \textbf{contact LGVI} produces a regular, almost monotone decrease of the
cost and quickly reaches a plateau, indicating a stable convergence of the
PMP iteration to a locally optimal control.  
The \textbf{RK2} discretisation, in contrast, exhibits a more oscillatory
behaviour: the cost decreases more slowly and typically settles at a slightly
higher value than the contact scheme.  

In other words, on this short horizon both methods do converge to reasonable
controls, but the LGVI--based PMP attains a marginally better objective and a
more stable iteration, while the RK2--based PMP shows a less regular cost
profile and tends to stagnate earlier.

The final optimal trajectories of each scheme are evaluated using the
geometric metrics introduced in Section~\ref{subsec:setup_metrics}.  
Here the difference between the two discretisations is much more pronounced.

For the \textbf{contact LGVI}, we observe that
\[
   \max_k \bigl|\tr(\rho_k)-1\bigr|
   \;\approx\; 10^{-15},
   \qquad
   \max_k \delta_{\mathrm{pos}}(k) \;=\; 0,
\]
i.e., trace and positivity are preserved up to machine precision along the
whole PMP trajectory.  
The contact--form defect
\[
   \theta_k 
     = (z_{k+1}-z_k)
       - \bigl\langle P_{k+1}, \rho_{k+1}-\rho_k\bigr\rangle
\]
remains at the expected $\mathcal{O}(\Delta t^2)$ scale, with typical values
around
\[
   \max_k|\theta_k| \;\sim\; 10^{-3}.
\]

For the \textbf{RK2} scheme, the optimal trajectories are no longer
geometrically consistent.  
In a representative run we obtain
\[
   \max_k \delta_{\mathrm{pos}}(k)
      \;\sim\; 10^{-3}, 
   \qquad
   \max_k|\theta_k| \;\sim\; 10^{-1},
\]
so that RK2 violates positivity and the discrete contact identity by
two to three orders of magnitude more than the contact LGVI, despite having
the same nominal order of accuracy.

Figures~\ref{fig:ocp_pulses_T3} and~\ref{fig:ocp_drift_T3} illustrate,
respectively, the resulting optimal control pulses and the geometric drifts
(contact--form defect and positivity drift) along the optimal trajectories.

\begin{figure}[h!]
  \centering
  \includegraphics[width=0.48\linewidth]{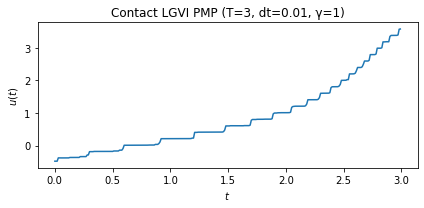}\hfill
  \includegraphics[width=0.48\linewidth]{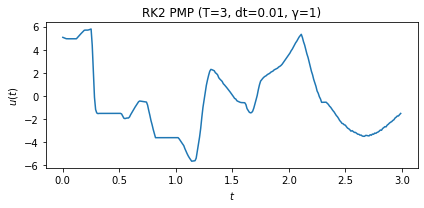}
  \caption{Optimal control pulses for $T=3$, $\Delta t=0.01$.
  Left: contact LGVI PMP; right: RK2 PMP.  
  The contact scheme yields smooth pulses and consistent convergence;
  the RK2 scheme oscillates and fails to reach a comparable optimum.}
  \label{fig:ocp_pulses_T3}
\end{figure}

\begin{figure}[h!]
  \centering
  \includegraphics[width=0.48\linewidth]{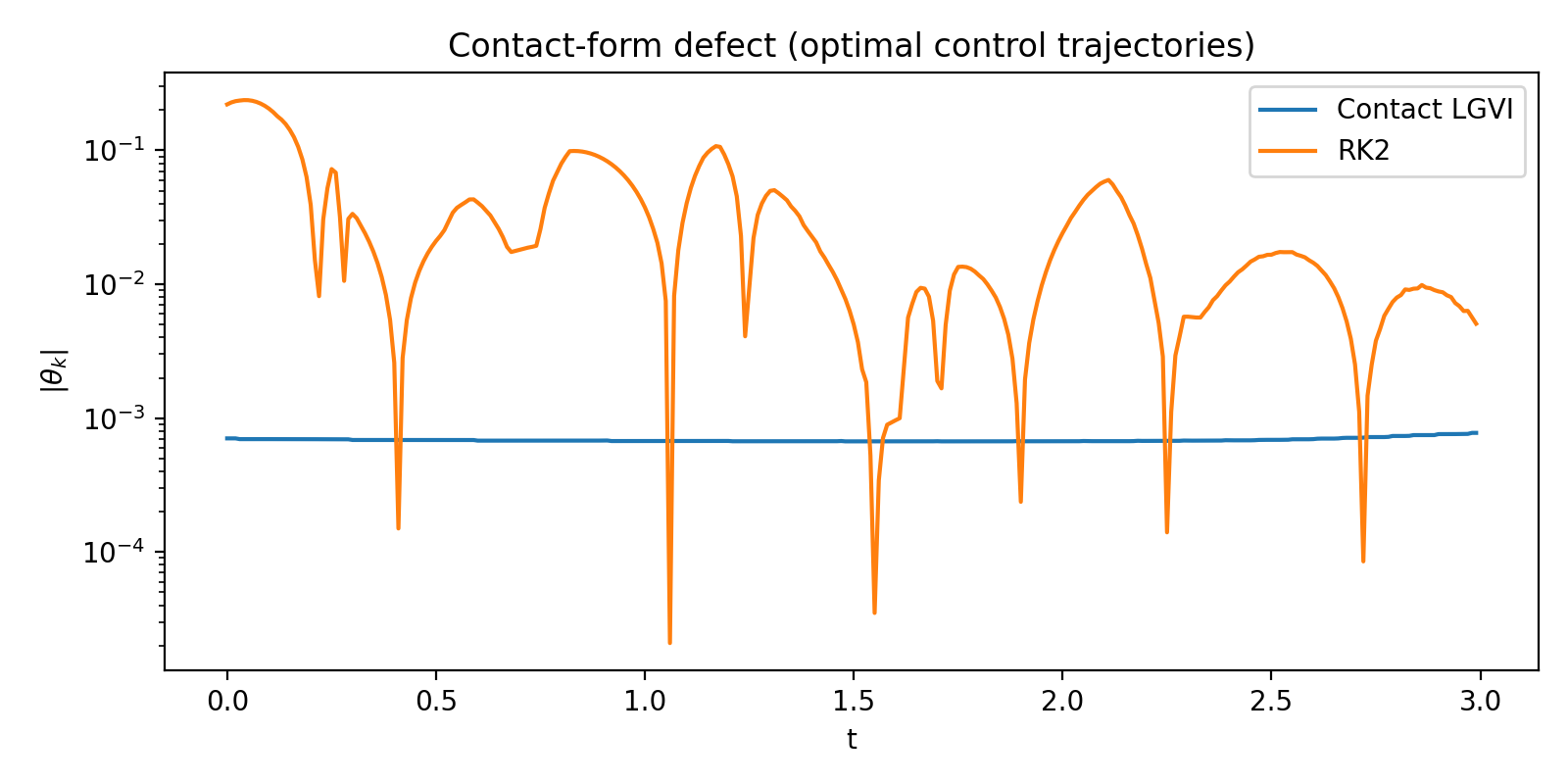}\hfill
  \includegraphics[width=0.48\linewidth]{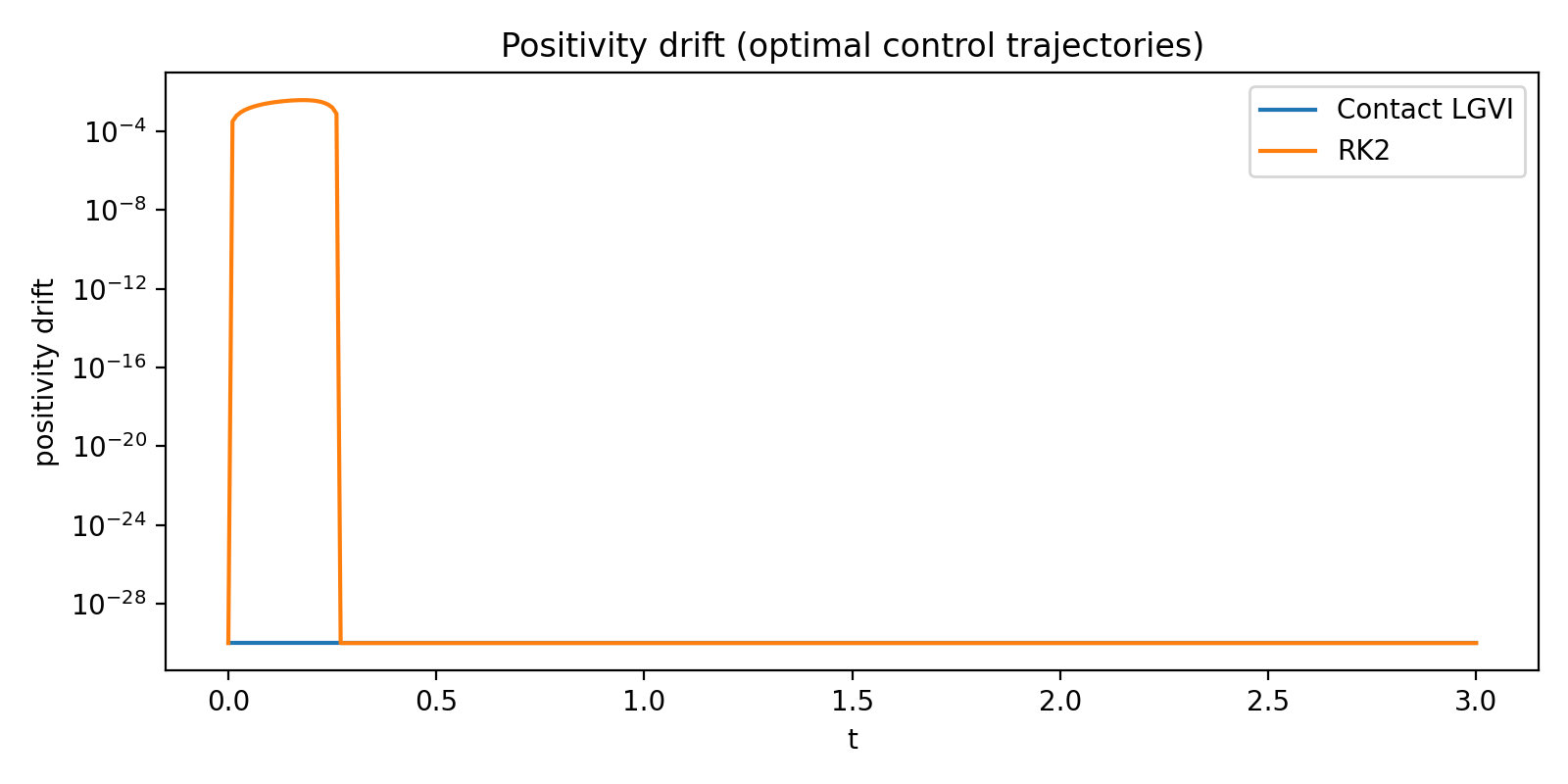}
  \caption{Geometric drifts along the optimal control trajectories.
  Left: contact-form defect $|\theta_k|$.  
  Right: positivity drift $-\min\{0,\lambda_{\min}(\rho_k)\}$.  
  The contact LGVI remains at machine precision and
  $\mathcal{O}(\Delta t^2)$; RK2 exhibits two to three orders of magnitude
  larger violations.}
  \label{fig:ocp_drift_T3}
\end{figure}

Overall, this experiment confirms that the discrete contact PMP equipped
with the contact LGVI provides a robust and geometrically consistent method
for optimal control of dissipative quantum systems.  
Even though both schemes are formally second order, the non-contact RK2
discretization fails to preserve CPTP structure and contact geometry,
and this degradation directly harms the convergence and reliability of the
PMP shooting iteration.

\section{Conclusions and outlook}
\label{sec:conclusions}

We have developed a discrete geometric framework for optimal control of open
quantum systems governed by Lindblad dynamics, based on the contact
formulation of the Pontryagin Maximum Principle and Lie--group variational
integrators.

At the continuous level, the PMP for systems with running costs naturally
lives on a contact manifold, with the state, costate, and cost variables
evolving according to a contact Hamiltonian flow.  
We have shown that this structure admits a faithful discrete counterpart:
starting from a discrete Herglotz principle, we constructed a discrete
contact PMP on manifolds, and then lifted it to Lie groups by means of
retractions and exact discrete contact Lagrangians.  
The resulting Lie--group contact variational integrators are generated by
type--II discrete contact Hamiltonians and yield strict discrete
contactomorphisms on the extended phase space.

We then specialized this construction to a single qubit subjected to
Hamiltonian control and amplitude--damping noise.  
For this Lindblad system we designed a second--order contact LGVI that
combines the exact Kraus map of the dissipative channel with the exact
unitary evolution through a symmetric Strang splitting.  
By construction, each time step is completely positive and trace--preserving
(CPTP) and preserves the discrete contact form associated with the cost
accumulation.  
This provides a fully geometric discretization of the Lindblad PMP: the
state, costate and cost variables evolve through an explicit discrete
contactomorphism generated by a type--II discrete generating function.

For comparison, we considered a standard explicit second--order
Runge--Kutta (Heun) integrator applied directly to the Lindblad generator.
This RK2 scheme has the same formal order of accuracy as the contact LGVI
but is not derived from a discrete Lagrangian or generating function, and
does not enforce CPTP structure nor discrete contact geometry.

Our numerical experiments illustrate the consequences of these geometric
differences.  
On short horizons with moderate damping, both schemes achieve comparable
state--space accuracy, but only the contact LGVI maintains machine--precision
trace and positivity, and exhibits a small contact--form defect consistent
with the expected $\mathcal{O}(\Delta t^2)$ behaviour.  
On long horizons and under strong dissipation, the contrast becomes
dramatic: the RK2 discretization develops catastrophic trace and positivity
violations, together with a contact--form defect that grows by many orders
of magnitude, while the contact LGVI remains numerically CPTP and
contact--consistent throughout.  
When embedded in a discrete PMP shooting algorithm for a qubit optimal
control problem, the contact LGVI leads to stable convergence towards
physically meaningful optimal pulses, whereas the RK2--based PMP suffers
from degraded performance and reduced reliability.

These results support the conclusion that contact--preserving geometric
integrators are not merely aesthetically appealing, but practically
advantageous for discrete optimal control of open quantum systems, especially
when long horizons, dissipation, and variational principles interact.

Several directions for future work are opened by this study.  
On the theoretical side, it would be natural to construct higher--order
contact LGVIs, to analyze their backward error and modified contact
structure, and to extend the discrete contact PMP to more general classes
of constraints and state--dependent running costs.  
On the physical side, applying the present framework to multi--qubit systems,
collective dissipation, and more complex noise models (e.g.\ dephasing or
finite--temperature baths) would further test its robustness and scalability.
Finally, combining contact LGVIs with advanced numerical optimization methods
(e.g.\ quasi--Newton or second--order techniques) and with experimental
constraints from quantum hardware constitute an interesting avenue towards
geometric, hardware--compatible optimal control for noisy intermediate--scale
quantum devices.


\begin{thebibliography}{99}

\bibitem{AltafiniTicozzi2012}
C.~Altafini and F.~Ticozzi,
``Modeling and control of quantum systems: An introduction,''
\emph{IEEE Transactions on Automatic Control}, 
vol.~57, no.~8, pp.~1898--1917, 2012.

\bibitem{Anahory2021JNS}
A.~Anahory Simoes, D.~Martín de Diego, M.~de León, and M.~Lainz.
On the geometry of discrete contact mechanics.
\emph{Journal of Nonlinear Science}, 31:88, 2021.

\bibitem{herglotzLG}
A.~Anahory, L.~J.~Colombo, M.~de León, J.~C.~Marrero, D.~M.~de Diego, and E.~Padrón,
Reduction by symmetries of contact mechanical systems on Lie groups.
\emph{SIAM Journal on Applied Algebra and Geometry}, 8(4), 821--845, 2024.

\bibitem{AssifChatterjeeBanavar2020}
M.~A.~P. Assif, D.~Chatterjee, and R.~Banavar.
A simple proof of the discrete-time geometric Pontryagin maximum principle on smooth manifolds.
\emph{Automatica}, 114:108847, 2020.

\bibitem{Maria}
M.~Barbero-Liñán and M.~C.~Muñoz-Lecanda.
Geometric approach to Pontryagin’s maximum principle.
\emph{Acta Applicandae Mathematicae} 108.2 (2009): 429--485.

\bibitem{BlanesCasas2005}
S.~Blanes and F.~Casas.
Splitting methods for nonautonomous separable dynamical systems.
\emph{Journal of Physics A: Mathematical and General}, 39(19):5405--5423, 2006.

\bibitem{Boltyanskii1978}
V.~G.~Boltyanskii,
\emph{Optimal Control of Discrete Systems},
John Wiley \& Sons, New York, 1978.

\bibitem{BonnardChybaSugny2009}
B.~Bonnard, M.~Chyba, and D.~Sugny,
Time-minimal control of dissipative two-level quantum systems: The generic case.
\emph{IEEE Transactions on Automatic Control}, 54(11):2598--2610, 2009.

\bibitem{BonnardCotsShcherbakovaSugny2010}
B.~Bonnard, O.~Cots, N.~Shcherbakova, and D.~Sugny,
The energy minimization problem for two-level dissipative quantum systems.
\emph{Journal of Mathematical Physics}, 51(9):092705, 2010.

\bibitem{BonnardSugny2009}
B.~Bonnard and D.~Sugny,
Geometric optimal control and two-level dissipative quantum systems.
\emph{SIAM Journal on Control and Optimization}, 48(3):1289--1308, 2009.

\bibitem{breuer}
H.-P.~Breuer and F.~Petruccione,
\emph{The Theory of Open Quantum Systems}.
Oxford University Press, 2002.

\bibitem{CelledoniOwren2014}
E.~Celledoni, H.~Marthinsen, and B.~Owren.
An introduction to Lie group integrators—basics, new developments and applications.
\emph{Journal of Computational Physics}, 257(B), 1040--1061, 2014.

\bibitem{ClarkBlochColomboRooneyCDC2017}
W.~Clark, A.~M.~Bloch, L.~J.~Colombo, and P.~Rooney,
Optimal control of quantum purity for $n=2$ systems.
In \emph{Proc.\ IEEE CDC}, pp.~1317--1322, 2017.

\bibitem{ClarkBlochColomboRooneyDCDSS2019}
W.~Clark, A.~M.~Bloch, L.~J.~Colombo, and P.~Rooney,
Optimal control of quantum purity for two-level dissipative systems.
\emph{DCDS-S}, pp.~2357--2372, 2019.

\bibitem{DeLeon2023JNS}
M.~de~León, M.~Lainz, and M.~C.~Muñoz-Lecanda.
Optimal control, contact dynamics and Herglotz variational problem.
\emph{Journal of Nonlinear Science}, 33(9):1--46, 2023.

\bibitem{GeigesContact}
H.~Geiges,
\emph{An Introduction to Contact Topology}.
Cambridge University Press, 2008.

\bibitem{GoriniKossakowskiSudarshan1976}
V.~Gorini, A.~Kossakowski, and E.~C.~G.~Sudarshan,
Completely positive dynamical semigroups of $N$‐level systems.
\emph{Journal of Mathematical Physics}, 17(5):821--825, 1976.

\bibitem{HairerLubichWanner2006}
E.~Hairer, C.~Lubich, and G.~Wanner,
\emph{Geometric Numerical Integration: Structure-Preserving Algorithms for Ordinary Differential Equations},
2nd ed., Springer Series in Computational Mathematics, Vol.~31,
Springer, Berlin, 2006.


\bibitem{IserlesMuntheKaasNørsettZanna2000}
A.~Iserles, H.~Z.~Munthe-Kaas, S.~P.~Nørsett, and A.~Zanna.
Lie-group methods.
\emph{Acta Numerica}, 9:215–365, 2000.

\bibitem{Kraus1983}
K.~Kraus,
\emph{States, Effects and Operations: Fundamental Notions of Quantum Theory}.
Lecture Notes in Physics, Vol.~190, Springer, 1983.

\bibitem{LeokShingel2012}
M.~Leok and T.~Shingel,
General techniques for constructing variational integrators.
\emph{Frontiers of Mathematics in China}, 7 (2012), 273--303.

\bibitem{Lindblad1976}
G.~Lindblad.
On the generators of quantum dynamical semigroups.
\emph{Communications in Mathematical Physics}, 48:119–130, 1976.

\bibitem{marrero}
J.~C.~Marrero, D.~Martín de Diego, and E.~Martínez,
Local convexity for second order differential equations on a Lie algebroid.
arXiv:2103.14418, 2021.

\bibitem{MarsdenWest2001}
J.\,E.~Marsden and M.~West,
Discrete mechanics and variational integrators.
\emph{Acta Numerica}, 10:357--514, 2001.

\bibitem{MuntheKaas1998}
H.~Munthe-Kaas,
Runge–Kutta methods on Lie groups.
\emph{BIT Numerical Mathematics}, 38(1):92--111, 1998.

\bibitem{Ohsawa2015Automatica}
T.~Ohsawa.
Contact geometry of the Pontryagin maximum principle.
\emph{Automatica}, 51:40--46, 2015.

\bibitem{PhogatChatterjeeBanavar2016}
K.~S.~Phogat, D.~Chatterjee, and R.~Banavar.
A discrete-time Pontryagin maximum principle on matrix Lie groups.
\emph{Automatica}, 68:207--216, 2016.

\bibitem{roca}
R.~T.~Rockafellar,
Lagrange multipliers and optimality.
\emph{SIAM Review}, 35(2):183--238, 1993.

\bibitem{Strang1968}
G.~Strang.
On the construction and comparison of difference schemes.
\emph{SIAM Journal on Numerical Analysis}, 5(3):506--517, 1968.

\bibitem{VermeerenContactVI}
M.~Vermeeren, A.~Bravetti, and M.~Seri.
Contact variational integrators.
\emph{Journal of Physics A}, 52(44):445206, 2019.

\bibitem{WisemanMilburn2009}
H.~M.~Wiseman and G.~J.~Milburn,
\emph{Quantum Measurement and Control},
Cambridge University Press, 2009.



\end{thebibliography}
\end{document}